\documentclass[svgnames]{sig-alternate-2013}

\usepackage{graphicx}
\usepackage{framed}
\usepackage{xspace}
\usepackage[defblank]{paralist}
\usepackage{amssymb}
\usepackage{amsmath}
\usepackage{tabularx}
\usepackage{todonotes}
\usepackage{multirow}
\usepackage{tikz}
\usepackage{url}
\usepackage{listings}
\usepgflibrary{arrows}

\setdefaultleftmargin{0cm}{0cm}{}{}{}{}

\renewenvironment{proof}{\noindent\emph{Proof.\ }}{\qed\par\smallskip}

\makeatletter
\newcommand{\nobrackettag}[0]{\def\tagform@##1{\maketag@@@{##1}}}
\makeatother

 \newcommand{\B}{\mathcal{B}}
\newcommand{\C}{\mathcal{C}}

 \renewcommand{\L}{\mathcal{L}}
\newcommand{\M}{\mathcal{M}} 
\renewcommand{\O}{\mathcal{O}} \renewcommand{\P}{\mathcal{P}}
 
\renewcommand{\S}{\mathcal{S}}

\newcommand{\tup}[1]{\langle #1\rangle}   
\newcommand{\set}[1]{\{#1\}}                      

\newcommand{\true}{\mathsf{true}}
\newcommand{\false}{\mathsf{false}}

\newcommand{\ISA}{\sqsubseteq}

\newcommand{\BOX}[1]{ [\!{-}\!] #1}
\newcommand{\DIAM}[1]{\langle\!{-}\!\rangle #1}

\newcommand{\inadom}{\ensuremath{\textsc{live}}}
\newcommand{\limp}{\rightarrow}

\newcommand{\uml}{\M}
\newcommand{\ocl}{\O}
\newcommand{\chartset}{\S}

\newcommand{\procset}{\P}

\newcommand{\cl}[1]{\texttt{#1}}
\newcommand{\op}[1]{\mathsf{#1}}
\newcommand{\oclq}[1]{\mathit{#1}}

\newcommand{\balsa}{BAUML\xspace}
\newcommand{\counter}{\C\xspace}
\newcommand{\mulpers}{\ensuremath{\mu\L_p}\xspace}

\newcommand{\comp}[3]{\textsc{cmp}_{#1}^{#2}(#3)}

\newcommand{\art}[1]{\textsc{artifacts}(#1)}
\newcommand{\artcl}[1]{\textsc{a-classes}(#1)}
\newcommand{\subart}[1]{\textsc{a-states}(#1)}
\newcommand{\parentart}[1]{\textsc{art}_{\cl{#1}}}
\newcommand{\chart}[1]{S_{#1}}
\newcommand{\eventset}{E}
\newcommand{\proc}[1]{P_{#1}}

\newcommand{\tasks}[1]{\textsc{tasks}(#1)}
\newcommand{\conds}[1]{\textsc{conditions}(#1)}

\newcommand{\slabel}[1]{\mathsf{#1}}

\newcommand{\ro}[1]{#1_r}
\newcommand{\rw}[1]{#1_{rw}}

\newcommand{\cval}[1]{\mathsf{#1}}

\newcommand{\tstate}[1]{\textsc{term}_{#1}}

\newcommand{\instr}[1]{\mathsf{#1}}
\newcommand{\halt}{\instr{HALT}}
\newcommand{\inc}[2]{\instr{INC}(#1,#2)}
\newcommand{\cdec}[3]{\instr{CDEC}(#1,#2,#3)}
\newcommand{\goto}[1]{\instr{GOTO}\ #1}

\newcolumntype{R}{>{\raggedleft\arraybackslash}X}
\newcolumntype{L}{>{\raggedright\arraybackslash}X}
\newcolumntype{C}{>{\centering\arraybackslash}X}

\newcommand{\defaulttablegraphics}[3]{
    \begin{minipage}{#2}\noindent

\vspace*{.05cm}

            \includegraphics[width=#2]{#1}

\vspace*{.05cm}
    \end{minipage}
}

\newcommand{\tableg}[2]{
        \defaulttablegraphics{#1}{#2}{-9cm}
}

\newcommand{\domr}[1]{#1|_1}
\newcommand{\imr}[1]{#1|_2}
\newcommand{\source}[2]{\textsc{src}_{#2}(#1)}
\newcommand{\target}[2]{\textsc{trg}_{#2}(#1)}

\newtheorem{exampleAux}{Example}[section]

\newtheorem{theorem}{Theorem}[section]

\newtheorem{corollary}[theorem]{Corollary}

\newfont{\mycrnotice}{ptmr8t at 7pt}
\newfont{\myconfname}{ptmri8t at 7pt}
\let\crnotice\mycrnotice%
\let\confname\myconfname%

\permission{This is the extended version of ``Verifiable UML
    Artifact-Centric Business Process Models'', to appear in the
  Proceedings of CIKM'14, Nov.~3--7, 2014, Shanghai, China.}
\copyrightetc{ACM \the\acmcopyr}
\crdata{http://dx.doi.org/10.1145/2661829.2662050}

\begin{document}

\title{Verifiable UML Artifact-Centric Business Process Models}

\numberofauthors{2}
\author{
%
\alignauthor
Diego Calvanese
\qquad
Marco Montali\\
      \affaddr{KRDB Research Centre for Knowledge and Data}\\
      \affaddr{Free University of Bozen-Bolzano}\\
      \email{\{calvanese,montali\}@inf.unibz.it}
\alignauthor
Montserrat Esta\~{n}ol
\qquad
Ernest Teniente\\
      \affaddr{Dept.\ of Services and Systems Engineering}\\
      \affaddr{Universitat Polit\`{e}cnica de Catalunya}\\
      \email{\{estanyol,teniente\}@essi.upc.edu}
}

\maketitle
\begin{abstract}
 Artifact-centric business process models have gained increasing momentum
  recently due to their ability to combine structural (i.e., data related) with
  dynamical (i.e., process related) aspects.  In particular, two main lines of
  research have been pursued so far: one tailored to business artifact modeling
  languages and methodologies, the other focused on the foundations for their
  formal verification. In this paper, we merge these two lines of research, by
  showing how recent theoretical decidability results for verification can be
  fruitfully transferred to a concrete UML-based modeling methodology. In
  particular, we identify additional steps in the methodology that, in
  significant cases, guarantee the possibility of verifying the resulting
  models against rich first-order temporal properties. Notably, our results can
  be seamlessly transferred to different languages for the specification of the
  artifact lifecycles.
\end{abstract}




\section{Introduction}
%
A business process consists of several activities performed in coordination in order to achieve a business goal \cite{Weske2007}. Since business processes are key to achieving an organization's goals, they should be free of errors and performed in an optimal way.

Traditional approaches to business process modeling have been based on a \emph{process}- or \emph{activity}-centric perspective, that is, they have tended to focus on the ordering of the activities that need to be carried out, underspecifying or ignoring the data needed by the process.

An alternative to activity-centric process modeling is the \emph{artifact}-centric (or \emph{data}-centric) approach. Artifact-centric process models represent both structural (i.e. the data) and dynamic (i.e. the activities or tasks) dimensions of the process. For this reason, they have grown in importance in recent years.
One of the research lines in this topic is focused on finding the best way of representing artifact-centric process models. Several graphical alternatives have been proposed, such as Guard-Stage-Milestone (GSM) models \cite{Hull2011b}, BPMN with data \cite{Meyer2013}, PHILharmonic Flows \cite{Kunzle2011} or a combination of UML and OCL \cite{Estanol2013}, to mention a few examples.

Despite this variety, it is important to guarantee the correctness of these models. In order to do so, a second line of research has focused on the foundations for the formal verification of artifact-centric business process models. The greatest part of these works represent the business process using models grounded on logic, such as Data-centric Dynamic Systems (DCDS) \cite{Hariri2013,Hariri2013b,Calvanese2012}. However, the problem with these models grounded on logic is that they are not practical at the business level as they are complex and difficult to understand by the domain experts.

In this paper, we merge these two lines of research, by showing how recent theoretical decidability results for verification can be fruitfully transferred to a concrete UML-based modeling methodology. In particular, we identify sufficient conditions over the models used by this methodology which guarantee decidabilty of verification. We also show how decidability of verification can be achieved when one of such conditions is not fulfilled. These results represent a significant step forward in the area since, to our knowledge, this is the first time that conditions for decidability are stated on models understandable by model experts, which are specified at a high level of abstraction.


As an aside result of our work, we identify a particular class of models, called \emph{shared instances}, characterized by the fact that there are two (or more) artifacts which share a read-only object. In this particular case, decidability is achieved by limiting the number of static objects with which an artifact can be related, by ensuring that all queries are navigational starting from the artifact and by imposing that no path of associations among two classes is navigated back and forth. Under these conditions, we can achieve decidability of verification without having to restrict reasoning over a bounded number of artifacts. More importantly, these results are not only applicable to our UML and OCL models, but can be extended to other languages for artifact-centric process models that fulfill the same conditions.

The rest of the paper is structured as follows. Section 2 introduces our framework. Section 3 presents our example, over which we will show the decidability conditions. Section 4 reports the decidability results. Finally, Sections 5 and 6 review the related work and present our conclusion.


\section{The \balsa Framework}

To facilitate the analysis of artifact-centric business process models, we base our work on the BALSA framework \cite{Hull2008}. It establishes four different dimensions that should be present in any artifact-centric business process model:

\begin{compactitem}
\item \textbf{Business Artifacts:} Business artifacts represent the data that the business requires to achieve its goals. They have an identifier and may be related to other business artifacts. One way of representing business artifacts is by using an entity-relationship model or a UML class diagram. Both diagrams are able to represent the business artifacts, their relationships and establish constraints on both.
\item \textbf{Lifecycles:} Lifecycles are used to represent the evolution of an artifact during its life, from the moment it is created until it is destroyed. Intuitively, they can be graphically represented by means of statecharts or state machine diagrams.
\item \textbf{Services:} Services are atomic units of work in the business process. They are in charge of evolving the process. As such, they make changes to artifacts by creating, updating and deleting them. They may be represented in different ways: alternatives range from using natural language to logic or with operation contracts specified in OCL.
\item \textbf{Associations:} Associations establish restrictions on the way services may change artifacts, that is, they impose constraints on services. They may be represented using a procedural representation, such as a workflow or BPMN, or using a declarative representation, such as condition-action rules.
\end{compactitem}
In contrast to artifacts, whose evolution we wish to track, in many instances, businesses need to keep data in the system that does not really evolve. In order to distinguish this data from artifacts, we will refer to it as \emph{objects}.

In this paper, we adopt the instantiation of BALSA in \cite{Estanol2013,Estanol2013b}, representing the aforementioned dimensions using UML \cite{UML2-4-1Super} and OCL \cite{OCL2-4}. Both UML and OCL are standard languages generally used for, but not limited to, conceptual modeling. In particular, we use: UML class diagrams for artifact, object, and relationship types; UML state machine diagrams for artifact lifecycles; UML activity diagrams for associations, and OCL operation contracts for services.
We call this concrete modelling approach \emph{BALSA UML} (\balsa for short).
However, this does not restrict our result to this subset of diagrams: the results are extendable to the rest of alternatives.

Technically, we define a \balsa model $\B$ as a tuple $\tup{\uml,\ocl,\chartset,\procset}$, where:
\begin{compactitem}
\item  $\uml$ is a UML class diagram, in which some classes represent (business) artifacts. Given two classes $A$ and $B$, we say that $A$ is a $B$, written $A \ISA_\uml B$, if $A = B$ or $A$ is a direct or indirect subclass of $B$ in $\uml$. Furthermore, given a class $A$ and a (binary) association $R$ in $\uml$, we write $A =_\uml \exists R$ ($A =_\uml \exists R^-$ resp.) if $A$ is the domain of $R$ (image of $R$ resp.) according to $\uml$. We also denote by $\domr{R}$ and $\imr{R}$ the role names attached to the domain and image classes of $R$.
 We denote the set of artifacts in $\uml$ as $\art{\uml}$ and, when convenient, we use $\art{\B}$ interchangeably. Each artifact is the top class of a hierarchy whose leaves are subclasses with a dynamic behavior (their instances change from one subclass to another). Each subclass represents a specific state in which an artifact instance can be at a certain moment in time. We denote by $\artcl{\uml}$ ($\artcl{\B}$ resp.) the set of such subclasses, including the artifacts themselves. Given a class $\cl{S} \in \artcl{\uml}$, we denote by $\parentart{S}$ the class $\cl{S}$ itself if $\cl{S}$ is an artifact, or the class $\cl{A}$ if $\cl{A}$ is an artifact and $\cl{S}$ is a possibly indirect subclass of $\cl{A}$.
Given an artifact $\cl{A} \in \art{\uml}$, we denote by $\subart{\cl{A}}$ the set of leaves in the hierarchy with top class $\cl{A}$.
\item $\ocl$ is a set of OCL constraints over $\uml$.
\item $\chartset$ is a set of UML state transition diagrams, one per artifact in $\art{\uml}$. In particular, for each artifact $\cl{A} \in \art{\uml}$, $\chartset$ contains a state transition diagram $\chart{\cl{A}} = \tup{V,v_0,\eventset,T}$, where $V$ is a set of states, $v_0 \in V$ is the initial state, $\eventset$ is a set of events, and $T \subseteq V \times \eventset \times \mathit{OCL}_\uml \times V$ is a set of transitions between pairs of states, each labelled by an event in $\eventset$ and by an OCL condition over $\uml$. In particular, the states $V$ of $\chart{\cl{A}}$ exactly mirror the classes in $\subart{\cl{A}}$, so that $\chart{\cl{A}}$  encodes the allowed event-driven transitions of an artifact instance of type $\cl{A}$ from the current state to a new subclass (i.e., a new artifact state). Moreover, the initial transitions leading to $v_0$ always result in the creation of an instance of the artifact being specified by $\chart{\cl{A}}$.
\item $\procset$ is a set of UML activity diagrams, such that for every transition diagram  $\tup{V,v_0,\eventset,T} \in \chartset$, and for every event $\varepsilon \in \eventset$, there exists one and only one activity diagram $\proc{\varepsilon} \in \procset$. With a slight abuse of notation, given a state transition diagram $\chart{} \in \chartset$, we denote by $\procset_{\chart{}} \subseteq \procset$ the set of activity diagrams referring to all events appearing in $\chart{}$.
\end{compactitem}
 In this paper, we will not impose any restriction on the control-flow structure of such activity diagrams, but only on their atomic tasks and conditions. For this reason, given an artifact $\cl{A} \in \art{\uml}$, we respectively denote by $\tasks{\cl{A}}$ and $\conds{\cl{A}}$ the set of atomic tasks and conditions appearing in the state transition diagram $\chart{a}$, also considering all activity diagrams related to $\chart{\cl{A}}$. We then define $\tasks{\B} = \bigcup_{\cl{A} \in \art{\uml}} \tasks{\cl{A}}$ and $\conds{\B} = \bigcup_{\cl{A} \in \art{\uml}} \conds{\cl{A}}$. Moreover, we assume that every task in $\tasks{\cl{A}}$ that does not belong to the activity diagram of an initial transition takes in input an instance of the artifact type in $\chart{a}$ and that every condition in $\conds{\cl{A}}$ is in the scope of such artifact.


In \balsa, conditions are expressed as OCL queries over the UML class diagram $\uml$. Similarly, each (atomic) task is associated to a so-called \emph{operation contract}, which expresses a precondition on the executability of the task, and a postcondition describing the effect of the task, both formalized in terms of OCL queries over $\uml$. The semantics of the operation contract is that the task can only be executed when the current information base satisfies its precondition, and that, once executed, the task brings the information base to a new state that satisfies the task postcondition. In this light, tasks represent services in the terminology of BALSA.

\section{Example}
\label{sec:example}




We present a relevant example based on a system for a company that registers orders from customers, and stores information about the orders made by the company to its suppliers. Our example is likely to specify a simplified version of the artifact-centric process models of an online shop like Amazon. We use a set of UML diagrams and OCL operation contracts to represent the example in a modeler-friendly way according to the BALSA framework.  

Figure \ref{fig:shared_objects} shows the class diagram that represents the business artifacts and classes of our example. Artifacts are characterized by having several substates, represented as subclasses, with a disjointness constraint. This constraint is necessary to ensure that each artifact evolves correctly. In addition, artifacts have a lifecycle, in our case represented as a UML state machine.

\begin{figure}
\begin{minipage}[t]{\columnwidth}
\vspace{0pt}
\includegraphics[width=\columnwidth]{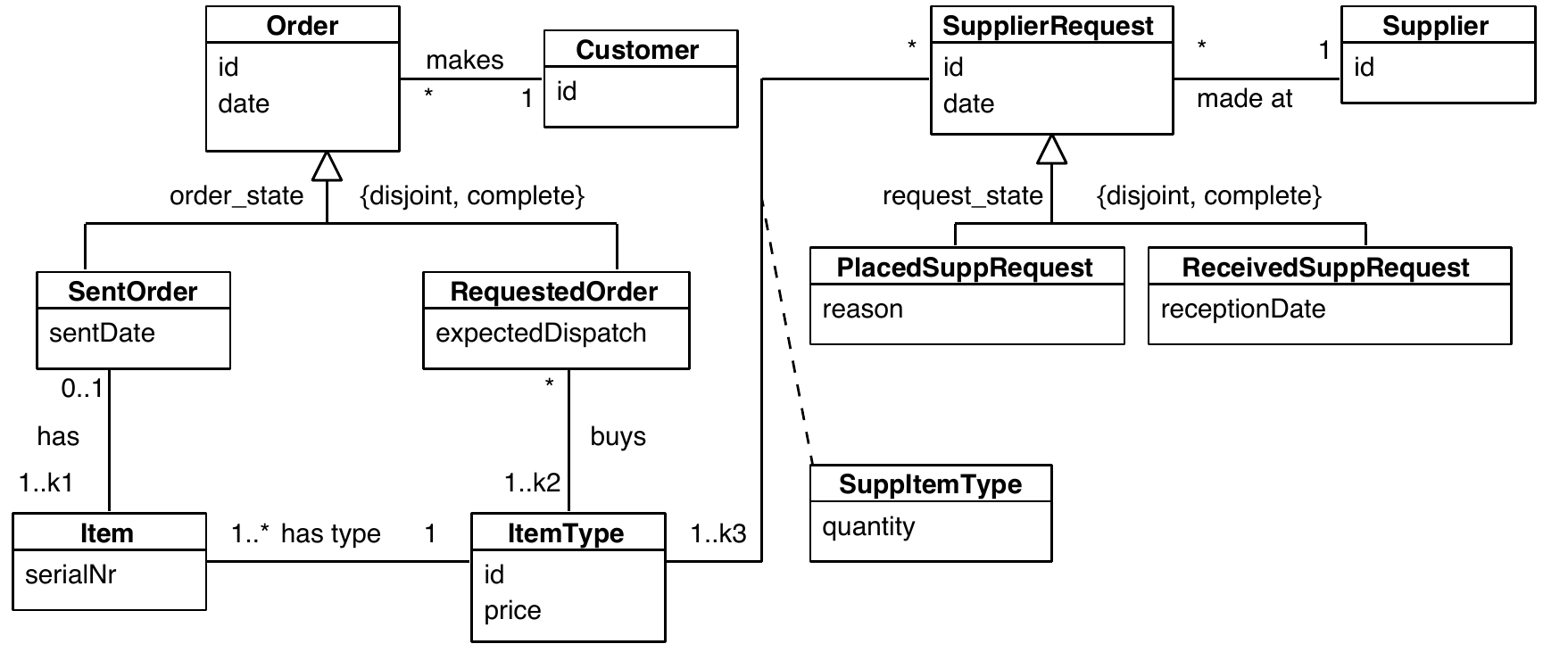}
\end{minipage}
\begin{minipage}[t]{\columnwidth}
\vspace{0pt}
\small
\scriptsize
Key constraints: \emph{serialNr} for \cl{Item}, \emph{id} for the other classes.
\end{minipage}
\caption{Class diagram for our example}
\label{fig:shared_objects}
\end{figure}

Logically, business artifacts, objects and associations to which they participate are created, updated, and destroyed by executing different services or tasks. However, we assume that some classes/associations are \emph{read-only}, i.e., their extension is not changed by the processes. This is the case, e.g., for \cl{ItemType} in our example.
%
%
Figure~\ref{fig:shared_objects} shows two business artifacts: \cl{Order} and \cl{SupplierRequest}. \cl{Order} has two substates: \cl{RequestedOrder} and \cl{SentOrder}, that track the order's evolution. A \cl{RequestedOrder} is related to various \cl{ItemType}s, indicating the products that the customer wishes to purchase.
On the other hand, \cl{SentOrder} is related to \cl{Item}s, which have a certain \cl{ItemType}. That is, \cl{SentOrder}s are directly related to specific items identified by their serial number. Notice that apart from the artifact itself, the associations \emph{makes}, \emph{has}, and \emph{buys} in which it takes part, are also created and deleted by the process.

Similarly, \cl{SupplierRequest} represents the requests made to the supplier. It has two possible substates: \cl{PlacedSuppRequest} and \cl{ReceivedSuppRequest}, and it is related to \cl{ItemType}, the association class that results from this relationship states information about the quantity of items of a certain type that have been requested to the supplier.

We call the artifact structure in this example \emph{shared instances by two artifacts} (\emph{shared instances} for short), because multiple artifacts (even of a different type) can be associated to the same object.
While an object might be related to an arbitrary number of artifact instances, the opposite does not hold, i.e., we naturally model an upper bound on the number of objects to which an artifact instance is related, so as to control the amount of information attached to the same artifact instance (e.g., consider the cardinalities of the \emph{buys} association in Figure~\ref{fig:shared_objects}).




\begin{figure}
\scriptsize \cl{Order}:\\
\hspace*{.05cm}
\includegraphics[height=.64cm]{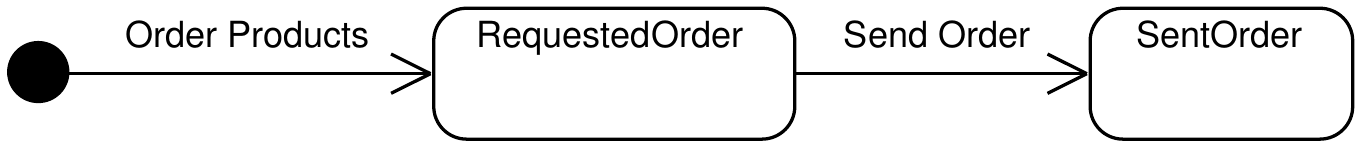}\\
\cl{SupplierRequest}:\\
\hspace*{.025cm}
\includegraphics[height=.64cm]{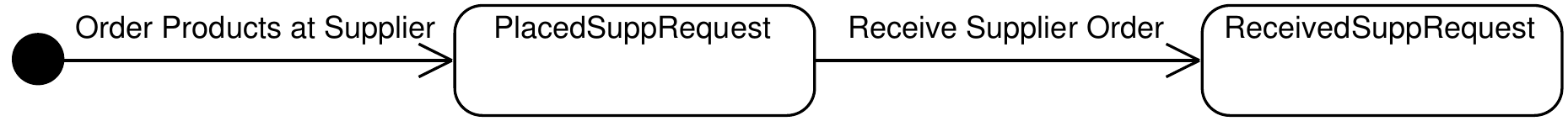}
\caption{Artifact state machines} 
\label{fig:shared_objects-smd-order}
\end{figure}


Both artifacts \cl{Order} and \cl{SupplierRequest} evolve independently from each other,
with a lifecycle specified by the state machines of Figure~\ref{fig:shared_objects-smd-order}.
Their meaning is very intuitive. In the case of \cl{Order}, when event \emph{Order Products} takes place, the \cl{RequestedOrder} is created. When we have a requested order and event \cl{Send Order} executes, the order is sent to the customer and the artifact changes its state to \cl{SentOrder}. The state machine diagram for \cl{SupplierRequest} is analogous to that of \cl{Order}.

Each of the events in the lifecycle transitions (\emph{Order Products}, \emph{Send Order}, \emph{Order Products at Supplier} and \emph{Receive Supplier Order}) are further defined using an activity diagram, which shows the units of work (i.e., the tasks) that are carried out, together with their execution order.

\begin{figure}
\includegraphics[width=\linewidth]{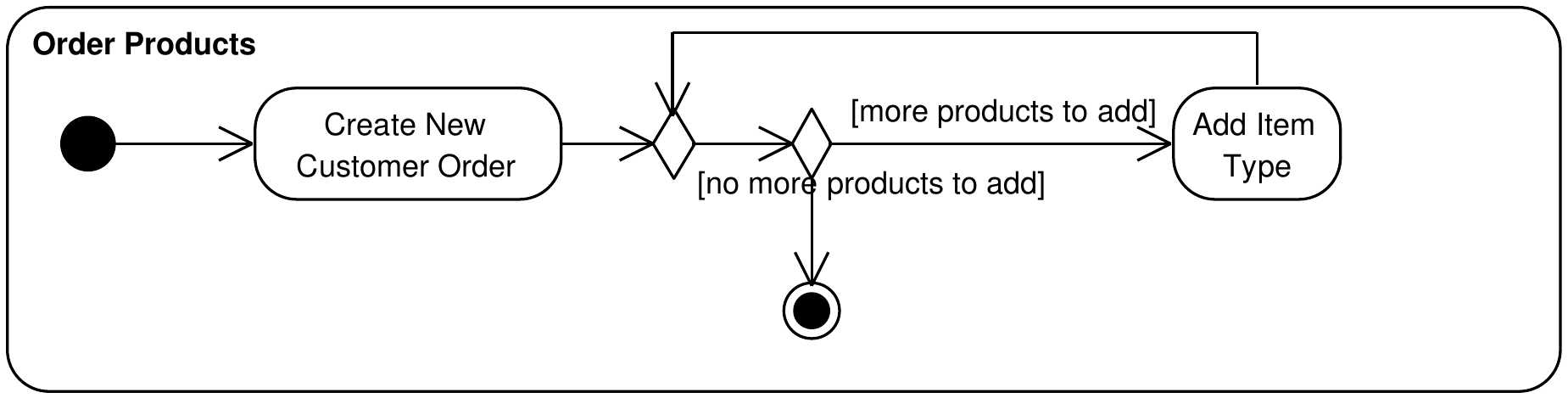}
\\
\includegraphics[width=0.52\linewidth]{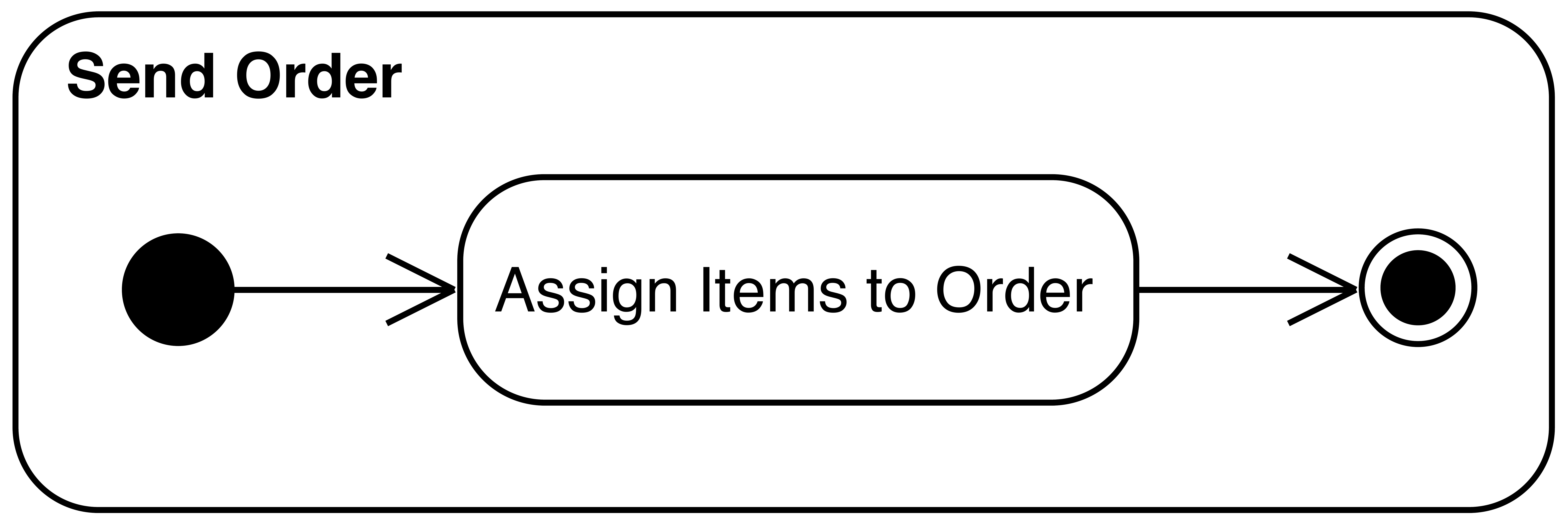}
\caption{Activity diagrams for the events of \cl{Order}}
\label{fig:shared_objects-ad-order_product}
\end{figure}

Figure~\ref{fig:shared_objects-ad-order_product} shows the activity diagrams for the events of \cl{Order}. As for the \emph{Order Products} event, the first task creates a new order, and the second task, which can be executed many times, adds an item type to the order that has been previously created.  As for the \emph{Send Order} event, the task adds the items to the order, marking it as sent.



Each activity diagram only gives an intuitive idea of what each task does. In order to specify tasks formally, we use OCL operation contracts, each of which has a precondition and a postcondition.
Below we show the OCL operation contracts for the tasks in Figure~\ref{fig:shared_objects-ad-order_product}.
\\[5pt]
$
\scriptsize
\begin{array}{@{}l@{}}
\textbf{action}~\op{CreateNewCustomerOrder}\\
~~(\mathit{orderId}:\cl{String}, \mathit{date}:\cl{Date}, \mathit{expDisp}:\cl{Date}):\cl{RequestedOrder}
\\
\textbf{pre:}~
\oclq{\neg \cl{RequestedOrder}.allInstances() \limp exists(ro | ro.id=orderId)}
\\
\textbf{post:}~
\oclq{\cl{RequestedOrder}.allInstances()} \\
\qquad
\begin{array}{@{}l@{}l@{}}
\oclq{\limp exists(ro |} & \oclq{ro.oclIsNew() \land ro.id=orderId \land ro.date=date}\\
&
\oclq{{}\land ro.expectedDispatch=expDisp \land result=ro)}
\end{array}
\end{array}
$\\


\noindent $\op{CreateNewCustomerOrder}$ receives as input the necessary parameters to create a new instance of the artifact \cl{RequestedOrder}. Its precondition makes sure that no other order with the same identifier exists. 
 It returns the \cl{RequestedOrder} that has been created with the input parameters.
This is an example of the kind of operation that is used to create new artifact instances, and that is typically associated to transitions leading to the initial artifact state (\emph{Order Products} in this example).\\[5pt]
$\scriptsize
\begin{array}{@{}l@{}}
\textbf{action}~\op{AddItemType}(\mathit{idItemType}:\cl{String}, \mathit{ro}:\cl{RequestedOrder})
\\
\textbf{pre:}~
\oclq{\neg ro.itemType.id \limp includes(idItemType)}
\\
\textbf{post:}~\oclq{ro.itemType.id \limp includes(id)}
\end{array}
$\\


\noindent $\op{AddItemType}$ adds an \cl{ItemType} to the order that has been created in the previous operation. Its precondition checks that the item type has not been already added to the order, and the postcondition creates the relationship between the given order and the right item type.

Notice that we assume that the artifact instance that is returned by the first operation, $\op{CreateNewCustomerOrder}$, is reused in the following operations. This assumption is necessary to ensure that we are always dealing with the same artifact instance.\\[5pt]
$\scriptsize
\begin{array}{@{}l@{}}
\textbf{action}~\op{AssignItemsToOrder}(\mathit{o}:\cl{RequestedOrder}, \mathit{date}:\cl{Date})
\\
\textbf{pre:}~
\begin{array}[t]{@{}l@{}l@{}}
\oclq{o.itemType \limp forAll(it |} & \oclq{it.item}\\
& \oclq{\limp exists(i | i.sentOrder\limp isEmpty()))}
\end{array}
\\
\textbf{post:}~
\begin{array}[t]{@{}l@{}}
\oclq{\neg o.oclIsTypeOf(\cl{RequestedOrder}) \land o.oclIsTypeOf(\cl{SentOrder}) \land{}}\\
\oclq{o.oclAsType(\cl{SentOrder}).sentDate=date\land{}}\\
\oclq{o.itemType\limp forAll(it | o.oclAsType(\cl{SentOrder}).item}\\
\oclq{\quad \limp includes(it.item@pre} \\
\oclq{\qquad\limp select(i | i.sentOrder \limp isEmpty()).asOrderedSet()}\\
\oclq{\quad\qquad \limp first()))}
\end{array}
\end{array}
$\\


\noindent $\op{AssignItemsToOrder}$ checks whether for the given \cl{RequestedOrder} $o$ it is the case
that there are available items (i.e., that have not been assigned to a \cl{SentOrder}) for each of the requested item types. If so, $o$ becomes a $\cl{SentOrder}$ that is associated to an available item for each of the requested item types.



These operation contracts show that the only elements that are created are the artifact itself and its relationships to other objects. Notice again that class \cl{ItemType}, which is shared by \cl{Order} and \cl{SupplierRequest}, is never modified by the tasks, and is in fact read-only. Moreover, all the actions that are not attached to the initial transition take as input an instance of the artifact type whose evolution is being modelled in the corresponding state machine, as required by our methodology. Notice that the navigation of all the OCL expressions in the pre and postconditions starts from the instance of the artifact flowing in the state transition diagram: an \cl{Order} (or one of its subclasses) in our example.

Given a \balsa model, it is interesting to check that it fulfills desired properties that ensure its correctness, such as the \emph{artifact termination} property: once an artifact instance is created, it should eventually evolve to a terminal state. This is addressed in the next section.

\section{Verification of \balsa Models}

The purpose of this section is to carefully analyze the interaction
between the dynamic and static component of \balsa models, so as to
single out the various sources of undecidability when it comes to
their verification. We show in particular that all the restrictions we
introduce towards decidability of
verification are in fact required: by
relaxing just one of them, verification becomes again undecidable.

\subsection{Verification Logic}
To specify temporal properties over \balsa models, we adopt the logic
$\mulpers$, a variant of first-order $\mu$-calculus that has been recently introduced to specify
requirements about the evolution of data-aware processes, jointly
considering the temporal dimension as well as the data maintained in
the different system states \cite{Hariri2013}. 
We recap here the main aspects of $\mulpers$ \cite{Hariri2013}, contextualizing it to
the case of \balsa models.

Given a \balsa model $\B$, the logic $\mulpers$ is defined as:
\begin{align*}
  \Phi ::=\ &
  Q ~\mid~
  \lnot \Phi ~\mid~
  \Phi_1 \land \Phi_2 ~\mid~
  \exists x. \inadom(x) \land \Phi ~\mid~ \\ &
  \inadom(\vec{x}) \land \DIAM{\Phi} ~\mid~
  \inadom(\vec{x}) \land\BOX{ \Phi} ~\mid~
  Z ~\mid~
  \mu Z.\Phi
\end{align*}
where $Q$ is a possibly open FO query, $Z$ is a
second order predicate variable,  and the following assumption holds:
in
$\inadom(\vec{x}) \land \DIAM{ \Phi}$ and $\inadom(\vec{x}) \land \BOX{\Phi}$,
the variables $\vec{x}$ are exactly
the free variables of $\Phi$, once we substitute to each
bounded predicate variable $Z$ in $\Phi$ its bounding formula $\mu
Z.\Phi'$ \cite{Hariri2013}.
This requirement expresses that $\mulpers$ quantifies
only over those objects/artifacts that persist in the system, i.e.,
continue to stay in the active domain of the system.

We make use of the following abbreviations:
\begin{compactitem}
\item $\Phi_1 \lor \Phi_2 = \neg (\neg\Phi_1 \land \neg \Phi_2)$,
\item $\BOX{\Phi} = \neg \DIAM{\neg \Phi}$,
\item $\nu Z. \Phi = \neg \mu Z. \neg \Phi[Z/\neg Z]$,
\item $\forall x.\cl{A}(x) \limp
 \Phi = \neg (\exists x.\cl{A}(x) \land \neg \Phi)$,
\item $\inadom(\vec{x}) \limp \DIAM{\Phi} = \neg
(\inadom(\vec{x}) \land \BOX{\neg \Phi})$,
\item $\inadom(\vec{x}) \limp \BOX{\Phi} = \neg
(\inadom(\vec{x}) \land \DIAM{\neg \Phi})$.
\end{compactitem}
The last two abbreviations show that $\mulpers$ allows one to
``control'' what happens when quantification ranges over a value that
disappears from the current active domain: in the $\limp$ case the
property trivializes to \emph{true}, in the $\land$ case it
trivializes to \emph{false}.


Among the properties of interest for \balsa models, we consider in
particular the fundamental requirement of \emph{artifact
  termination}. Intuitively, this property states that in all possible
evolutions of the system, whenever an artifact instance of a certain
type is present in the system, it must persist in the system until it
eventually reaches (in a finite amount of computation steps) a proper termination state.
Remember that such a state will have a counterpart in the UML
model of $\B$, which will contain a subclass for that specific
state. By denoting with $\tstate{\cl{A}} \in \subart{\cl{A}}$ the
proper termination state of artifact $\cl{A} \in \art{\B}$, and by
considering the standard FOL encoding of UML classes as unary
predicates, the
artifact termination property can be formalized in $\mulpers$ as
follows:
\[
\nu Z. \Big(\bigwedge_{\cl{A} \in \art{\B}} \hspace*{-.6cm}(\forall x. \cl{A}(x) \limp
\mu Y. \tstate{\cl{A}}(x) \lor (\cl{A}(x) \land \DIAM{Y}))\Big) \land \BOX{Z}
\]

In the following, all the undecidability results we give do not only
hold for the $\mulpers$ logic in general, but specifically for the
artifact termination property. Furthermore, we do only consider data
coming from a countably infinite unordered domain, and that can only
be compared for (in)equality. We thus avoid any
assumption on the structure of data domains, and consider only string
and boolean attributes\footnote{A boolean attribute can be considered
  as a special string attribute that can only be assigned to the
  special strings $\cval{true}$ or $\cval{false}$. This constraint can
be easily expressed in OCL.}.
In this light, our results witness that it is not
possible to achieve meaningful restrictions towards decidability just
by restricting the property specification logic, but that it is instead
necessary to suitably restrict the expressiveness of \balsa models themselves.

Since all the undecidability proofs rely on the encoding of 2-counter
machines \cite{Min67} into the specific class of \balsa
models under analysis, we start by briefly recapping 2-counter machines.

\subsection{2-Counter Machines}
We follow the original formulation in \cite{Min67}. A \emph{counter} is a memory
register that stores a non-negative integer. Given two positive integers $n, m \in \mathbb{N}^+$, an
\emph{$m$-counter machine} $\counter$ with counters $c_1,\ldots,c_m$
is a program with $n$ commands:
\[
1: \mathit{CMD}_1; \quad 2: \mathit{CMD}_2; \quad \ldots \quad n: \halt;
\]
where each $\mathit{CMD}_k$ (for \emph{index} $k \in \set{1,\ldots,n-1}$) is either
an increment command or a conditional decrement command.

Given $i \in \set{1,\ldots,m}$, an \emph{increment
   command} for counter $i$, written $\inc{i}$, is a command that
 increases the counter $c_i$ of one unit, and then
  jumps to the next instruction. Formally, for $k,k' \in \{1,\ldots,n-1\}$,
\[
k: \inc{i}{k'} \quad \text{ means } \quad k: c_i:=c_i+1;~\goto{k'};
\]

Given $i \in \set{1,\ldots,m}$, and $k,k',k''
  \in \set{1,\ldots,n}$, a \emph{conditional decrement instruction}
  for counter $i$ and instruction $k$, written $\cdec{i}{k'}{k''}$, tests
  whether the value of counter $i$ is zero. If so, it jumps to
  instruction $k'$; otherwise, it decreases counter $i$ of one unit,
  and then jumps to instruction $k''$. Formally, for $k,k',k'' \in \{1,\ldots,n-1\}$,
command $k: \cdec{i}{k'}{k''}$ means
\[k:
\textbf{if}~c_i = 0 ~\textbf{then}~\goto{k'};  \textbf{else}~\{ c_i:=c_i-1;~\goto{k''}; \}\\
\]
An \emph{input} for an $m$-counter machine is an $m$-tuple
$\tup{d_1,\ldots,d_m}$ of values in $\mathbb{N}$
initializing its counters. Given an $m$-counter machine $\counter$ and an
input $I$ of size $m$, we say that $\counter$ \emph{halts on input}
$I$ if the execution of $\counter$ with counter initial values set by
$I$ eventually reaches the last, $\halt$ command.

It is well-known that checking whether a 2-counter machine halts on a
given input is undecidable \cite{Min67}, and it is easy to strengthen
this result as follows:


\begin{corollary}
\label{cor:2cm}
It is undecidable to check whether a 2-counter machine halts on
input $\tup{0,0}$.
\end{corollary}

 \begin{proof}
 Given a 2-counter machine $\counter$ with input $I = \tup{d_1,d_2}$,
 one can produce a new 2-counter machine $\counter'$ whose program is constituted by the
sequence of the following instruction sets:
\begin{compactenum}
\item a series of $d_1$ commands of type $k: \inc{1}{k+1}$;
\item a series of $d_2$ commands of type $k: \inc{2}{k+1}$;
\item the program of $\counter$, whose indexes are translated of
  $d_1+d_2$ units.
\end{compactenum}
It is easy to see that $\counter$ halts on input $\tup{d_1,d_2}$ if
and only if $\counter'$ halts on input $\tup{0,0}$.
\end{proof}
 In the following, we say that a 2-counter machine halts if it halts on
 input $\tup{0,0}$.

\subsection{Unrestricted Models}
We start by showing that, if we do not impose restrictions on the
shape of OCL queries used in
the pre-/post-conditions of tasks and in the decision points of a
\balsa model, then verification of artifact termination is undecidable. We say that a \balsa model is
\emph{unrestricted} if it does not impose any restriction on the shape of such queries.
\begin{theorem}
 \label{thm:undec-unrestricted}
Checking termination over unrestricted
\balsa models is undecidable.
\end{theorem}
\begin{proof}
By reduction from the halting problem of 2-counter machines, which is undecidable
(cf.~Corollary~\ref{cor:2cm}). Specifically, given a 2-counter machine
$\counter$, we produce a corresponding unrestricted \balsa model
$\B_\counter =
\tup{\uml^u,\emptyset,\set{\chart{\cl{2CM}}^u},\set{\proc{init}^u,\proc{run}^u}}$,
whose components are illustrated in Table~\ref{tab:2cm-unrestricted}. The idea behind the reduction is as
follows. $\uml$ contains a single artifact $\cl{2CM}$, which can be
ready or halted, the latter being the termination state
($\tstate{\cl{2CM}} = \cl{Halted2CM}$), as it can be clearly seen in
$\chart{\cl{2CM}}^u$.
As specified in diagram $\chart{\cl{2CM}}$, the $\op{init}$ operation
is activated only if the extension of  $\cl{Flag}$ is empty. In this
case, a new artifact instance of
  $\cl{Ready2CM}$ and a new object of type $\cl{Flag}$ are
  simultaneously created. The creation of a $\cl{Flag}$ object has the effect of blocking the possibility of creating new instances of
  $\cl{Ready2CM}$, in turn ensuring that only a single instance of
  $\cl{Ready2CM}$ will be created, and that only one execution of
  $\proc{run}$ will run. In fact, the only instance of
  $\cl{2CM}$ that enters  $\chart{\cl{2CM}}$ will move to the
  \emph{halted} state by executing the activity diagram
$\proc{run}$. In turn, $\proc{run}$ encodes the
program of $\counter$, by combining the process fragments obtained by
translating the single commands in $\counter$ as specified in
Table~\ref{tab:2cm-unrestricted}. Two classes
$\cl{Item}_1$ and $\cl{Item}_2$ are used to mirror the two
counters. In particular, at a given moment in time, the number of
instances of $\cl{Item}_i$ represents
the value of counter $i$. In this light:
\begin{inparaenum}[\it (i)]
\item incrementing counter $i$ translates into the creation of a new
instance of $\cl{Item}_i$;
\item testing whether counter $i$ is 0 translates into checking whether the extension of
  class $\cl{Item}_i$ is empty;
\item decrementing counter $i$ translates into the deletion of
  one of the current instances of $\cl{Item}_i$.
\end{inparaenum}
Table~\ref{tab:2cm-unrestricted} shows how these three aspects can be
formalized in terms of activity diagrams and OCL queries (focusing on
counter $1$). The diamond
gateways at the beginning of each fragment are used to properly merge
multiple incoming paths.

The claim follows by observing that $\counter$ halts if and only
if the unique instance of $\cl{2CM}$ that enters
$\chart{\cl{2CM}}$ also reaches the $\slabel{Halted2CM}$
state, i.e., properly terminates.
\begin{table*}[t]
\begin{tabularx}{\textwidth}{|c|@{}C|}
\hline
$\uml^u$
&
$\chart{\cl{2CM}}^u$
\\
\hline
\tableg{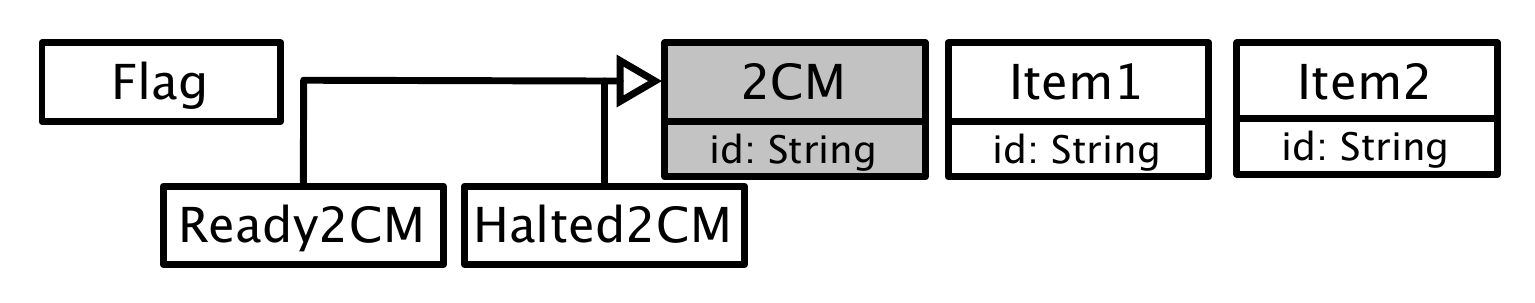}{8cm}
&
\tableg{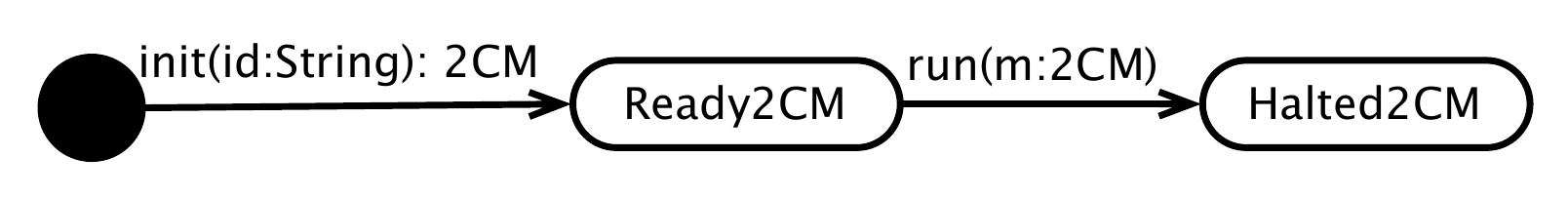}{8.5cm}
\\
\hline
\end{tabularx}

\begin{footnotesize}
\begin{tabularx}{\textwidth}{|l@{\ }|p{3.6cm}@{\ }|X|}
 \hline
 $P_{init}^u$
&
\tableg{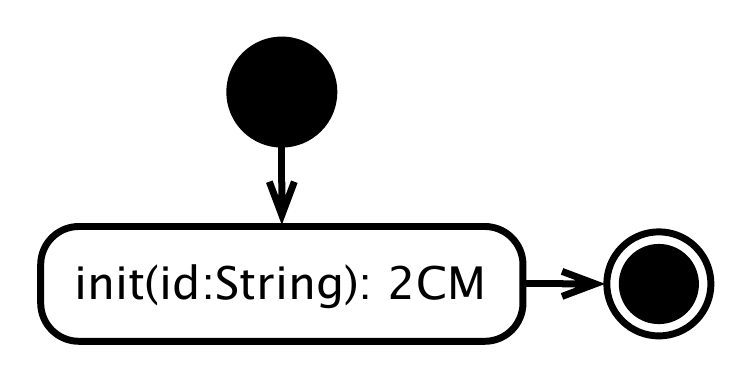}{3.8cm}
&
\vspace*{-.75cm}
$
\op{init}\
\begin{array}[t]{@{}l@{\ }l}
 \textbf{pre:}  & \oclq{\cl{Flag}.allInstances() \limp isEmpty()}\\
& \land \neg  (\cl{Ready2CM}.allInstances()\limp exists(m'|m'.id = id) ) \\
\textbf{post:} &
 \oclq{\cl{Flag}.allInstances()\limp exists(f|f.oclIsNew())}\\
&\oclq{\land \cl{Ready2CM}.allInstances()\limp exists(m|m.oclIsNew()
  \land m.id = id \land result
 =m)}
\end{array}
$
\\
\hline
\end{tabularx}
\end{footnotesize}

\begin{footnotesize}
\begin{tabularx}{\textwidth}{|l@{\ }|p{2.7cm}@{}|@{}p{3.9cm}|X|}
\hline
\multirow{4}{*}{$P_{run}^u$}
&
 start
&
\tableg{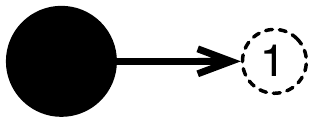}{1.6cm}
&
\\
\cline{2-4}
&
 $k: \inc{1}{k'}$
&
\tableg{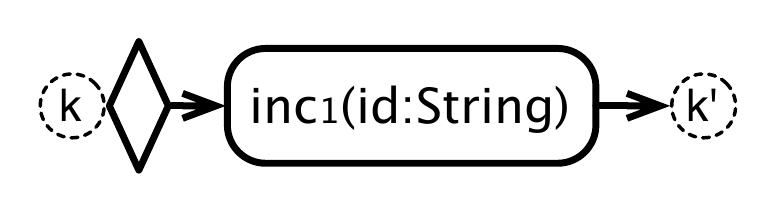}{4.1cm}
&
$
\op{inc_1} \
\begin{array}[t]{@{}l@{\ }l}
\textbf{pre:} & \oclq{\neg(\cl{Item1}.allInstances()\limp
  exists(i'|i'.id = id))}\\
\textbf{post:} & \oclq{\cl{Item1}.allInstances()\limp
  exists(i|i.oclIsNew() \land i.id = id)}
\end{array}
$
\\
\cline{2-4}
&
 $k: \cdec{1}{k'}{k''}$
&
\tableg{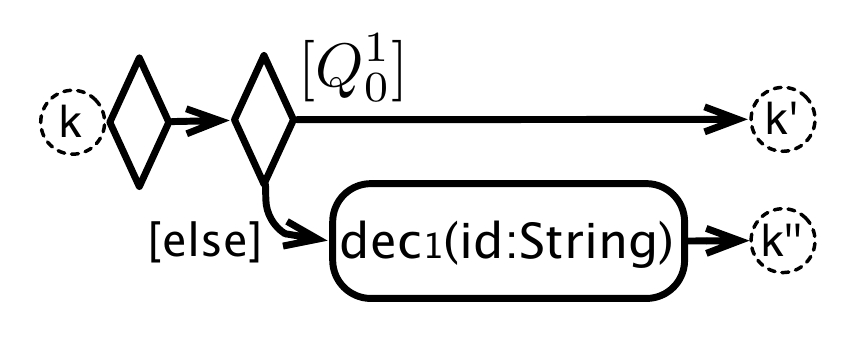}{4.1cm}
&
\vspace*{-.5cm}
$
\begin{array}{@{}l}
Q_0^1 = \oclq{\cl{Item1}.allInstances() \limp isEmpty()}\\[6pt]
\op{dec_1}\
\begin{array}[t]{@{}l@{\ }l}
 \textbf{pre:}  & \oclq{\cl{Item1}.allInstances() \limp
   exists(i | i.id = id)}\\
\textbf{post:} & \oclq{\neg(\cl{Item1}.allInstances()  \limp exists(i | i.id = id))}
\end{array}
\end{array}
$
\\
\cline{2-4}
&
$n: \halt$
&
\tableg{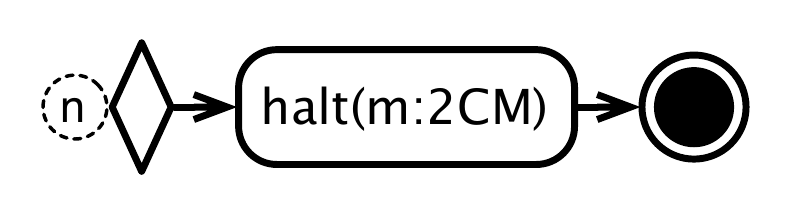}{4cm}
&
$
\op{halt} \
\begin{array}[t]{@{}l@{\ }l}
  \textbf{post:} & \oclq{\neg(m.oclIsTypeOf(\cl{Ready2CM}))}\\
&\oclq{\land m.oclIsTypeOf(\cl{Halted2CM})}
\end{array}
$
\\
\hline
\end{tabularx}
\end{footnotesize}
\caption{Unrestricted \balsa model simulating a
  2-counter machine}
\label{tab:2cm-unrestricted}
\end{table*}
\end{proof}

\subsection{Navigational and Unidirectional Models}
\label{sec:nav}
The proof of Theorem~\ref{thm:undec-unrestricted} relies on the
fact that artifact instances freely manipulate (i.e., create, read,
delete) instances of other
classes. Towards decidability, we have therefore to properly control
how artifact instances relate to other objects. In this light, we
suitably restrict OCL expressions, by allowing only
so-called navigational expressions.

To define navigational queries over a \balsa model  $\B = \tup{\uml,\ocl,\chartset,\procset}$, we start by
partitioning the associations and classes in $\uml$ into two sets: a read-only
set $\ro{\uml}$, and a read-write set $\rw{\uml}$. Intuitively,
$\ro{\uml}$ represents the portion of $\uml$ whose data are only
accessed, but never updated, by the execution of tasks, whereas
$\rw{\uml}$ represents the portion of $\uml$ that can be freely
manipulated by the tasks. These two sets can either be directly
specified by the modeler, or easily extracted by inspecting all
postconditions of operations present in $\procset$, marking a class
$\cl{C}$ as
read-write every time a sub-expression $\oclq{obj.oclIsNew()}$ appears
in some operation, and $obj$ is an instance of $\cl{C}$.
In this light, all artifacts presents in
$\uml$ are always part of the read-write set: $\art{\uml} \subseteq
\rw{\uml}$.

Given an object $obj$, an OCL expression is \emph{navigational from
  $obj$} if it is defined by means of the usual OCL operations like
exists, select, \ldots, but in which each subexpression is a boolean combination of expressions $Q_i$ that obey to
one of the following two types:
\begin{compactitem}
\item $Q_i$ only uses role and class names from $\ro{\uml}$;
\item $Q_i$ has the form of a path $o.r_1\cdots r_n$, which starts
  from $o$ and navigates through roles $r_1$ to $r_n$, where each
  $r_i$ is either a role or an attribute, and where $o$ is either the
  original object $obj$, or a variable used in the current operation.
\end{compactitem}
A \balsa model $\B = \tup{\uml,\ocl,\chartset,\procset}$
is \emph{navigational} if:
\begin{compactitem}
\item For every operation in $\procset$, with the exception of the
  $\op{init}$ operation, the OCL expressions used in its pre- and
  post-conditions are navigational from $a$, where $a$ is (the name of)
  the artifact instance taken in input by the operation.
\item Every condition in $\conds{\B}$ is an OCL expression that is
  navigational from (the name of) the artifact instance present in the
  scope of the condition.
\end{compactitem}
Navigational \balsa models do not
allow artifact instances to \emph{share} objects from read-write classes. Indeed, for an
artifact instance to establish a relation with an object of class
$\cl{C}$ previously created by another artifact instance, it is
necessary to write an OCL query that selects objects of type $\cl{C}$,
but this query is not navigational.

In spite of this observation, we will see that restricting \balsa models to navigational queries is
still not sufficient, but additional
requirements are needed towards decidability. The first requirement is
related to the way OCL expressions navigate the roles in
$\uml$. Given a navigational \balsa model $\B =
\tup{\uml,\ocl,\chartset,\procset}$, and given a role $r$ in $\uml$,
if there exists an OCL
expression in $\B$ that mentions $r$, then we say that $r$ is a
\emph{target role}, written $\target{r}{\B}$, otherwise we say that $r$ is a \emph{source
  role}, written $\source{r}{\B}$. We use this notion to define the notion of dependency between
two classes. Given classes $\cl{C}_1$ and $\cl{C}_{n+1}$ in $\uml$, we say
that $\cl{C}_{n+1}$ \emph{depends on} $\cl{C}_1$ if there exists a tuple
$\tup{A_1,\ldots,A_n}$ of binary associations such that each $A_i$
connects $\cl{C}_i$ and $\cl{C}_{i+1}$, and the role of $A_i$ attached
to $\cl{C}_{i+1}$ is a target role.
We then say that $\B$ is \emph{bidirectional} if it is navigational
and there exists a class in $\rw{\uml}$ that depends on itself or on
one of its super/sub-classes,
\emph{unidirectional} if it is navigational and there is no class in
$\rw{\uml}$ that depends on itself or on one of its super/sub-classes. Intuitively, for a unidirectional \balsa model
it is possible to mark each association in its UML
model as directed (since no association can have both nodes as targets), and the
resulting directed graph is acyclic. See for example
Table~\ref{tab:2cm-unbounded}. This property, in turn, can be tested
in $\textsc{NLogSpace}$.

We now correspondingly characterize navigation in
$\mulpers$. Without loss of generality, we consider only
binary relations\footnote{Non-binary relations can be removed
through reification.}. A \emph{pseudo-navigational} $\mulpers$ property has the form
\[
\begin{array}{@{}l@{}l@{}}
  \Phi ::=\ & \true \mid \false \mid A(x) \mid \lnot A(x) \mid \Phi_1 \land
  \Phi_2 \mid \Phi_1 \lor \Phi_2 \mid \\
&  Z \mid \mu   Z.\Phi \mid \nu Z.\Phi \mid\\
& \exists x. A(x) \land \Phi(x) \mid \forall x. A(x) \limp \Phi(x)
\mid \\
& \exists y. R(x,y) \land \Phi(y) \mid \forall y. R(x,y) \limp \Phi(y)
\mid \\
& \exists y. R(y,x) \land \Phi(y) \mid \forall y. R(y,x) \limp \Phi(y)
\mid \\
& A(x) \land \DIAM{\Phi}  \mid  A(x) \land\BOX{ \Phi} \mid
A(x) \limp \DIAM{\Phi}  \mid  A(x) \limp \BOX{ \Phi}
\end{array}
\]
where, in the last row, variable $x$ is exactly
the single free variable of $\Phi$, once we substitute to each
bounded predicate variable $Z$ in $\Phi$ its bounding formula $\mu
Z.\Phi'$ (resp., $\nu Z.\Phi'$). Notice that
pseudo-navigational properties are in negation normal form, and that
they constitute indeed a fragment of $\mulpers$. In fact, even if they
do not make use of $\inadom$, they always guard quantification and
next-state transitions with classes and/or relations, which imply the
corresponding quantified objects to be in the current active domain.

Given a unidirectional \balsa model $\B =
\tup{\uml,\ocl,\chartset,\procset}$, we characterize the fact
that a closed, pseudo-navigational $\mulpers$
property $\Phi$ is \emph{navigationally compatible} with
$\B$ as:
\begin{compactitem}
\item $\Phi$ contains a subformula of the form  $\exists x. A(x) \land
  \Psi(x)$ or $\forall x. A(x) \limp \Psi(x)$.
\item The largest subformula of $\Phi$ of the form
  $\exists x. A(x) \land \Psi(x)$ or $\forall x. A(x) \limp \Psi(x)$
  is such that:
\begin{inparablank}
\item $A \in \artcl{\B}$, and
\item $A$ and $x$ are
  compatible with $\Psi$, written $\comp{A}{x}{\Psi} = \true$,
  according to the notion of compatibility defined below.
\end{inparablank}
Given a class  $C$ in $\uml$, a variable $x$, and a
pseudo-navigational
  open $\mulpers$ property $\Phi(x)$, we define $\comp{C}{x}{\Phi}$ as:
\begin{compactenum}[(1)]
\item $\true$  if  $\Phi \in \set{\true, \false, Z}$
\item $C \ISA_\uml A \lor A \ISA_\uml C$ if $\Phi \in \set{A(x), \neg
  A(x)}$
\item $\comp{C}{x}{\Phi_1} \land \comp{C}{x}{\Phi_2}$ if $\Phi \in \set{\Phi_1 \land \Phi_2,\Phi_1 \lor \Phi_2}$
\item $\comp{C}{x}{\Psi}$ if $\Phi \in \set{\mu Z. \Psi, \nu
  Z. \Psi}$
\item $\false$ if $\Phi \in \set{\exists y. A(y) \land \Psi(y),\forall
    y. A(y) \limp \Psi(y)}$
\item $\target{\imr{R}}{\B} \land \comp{C'}{y}{\Psi} \land ((C
  \ISA_\uml \exists R) \lor (\exists R \ISA_\uml C))$\\ if
$\Phi \in \set{\exists y. R(x,y) \land \Psi(y), \forall y. R(x,y) \limp
 \Psi(y)}$\\ and $C' =_\uml \exists R^-$
\item $ \target{\domr{R}}{\B} \land \comp{C'}{y}{\Psi} \land ((C
  \ISA_\uml \exists R^-) \lor (\exists R^- \ISA_\uml C)) $
\\
if $\Phi \in \set{\exists y. R(y,x) \land \Psi(y), \forall y. R(y,x) \limp
 \Psi(y)} $ \\and $C' =_\uml \exists R$
\item $(C \ISA_\uml A \lor A \ISA_\uml C) \land \comp{C}{x}{\Psi} $\\ if $\Phi \in \{ A(x)
 {\land} \DIAM{\Psi}, A(x)
 {\land}\BOX{ \Psi}, A(x)
 {\limp}\DIAM{ \Psi}, A(x)
 {\limp}\BOX{ \Psi}\}$
\end{compactenum}

\end{compactitem}

Intuitively, the formulae above state that:
(1) $C$ and $x$ are always compatible with
non-first-order subformulae.
(2) $C$ and $x$ are
compatible with first-order components of the form $A(x)$ or $\neg
A(x)$ if classes $A$ and $C$ belong to the same hierarchy
according to $\uml$; this means that navigation through classes is only
allowed in the context of the same hierarchy.
(3) boolean connectives distribute the compatibility
  check to all their inner sub-formulae.
(4) fixpoint constructs push the compatibility check
  to their inner sub-formulae.
(5) compatibility is broken if new quantified
  variables over classes are introduced in the formula. This means
  that at most one quantification over classes is allowed in a
  pseudo-navigational property to be navigationally compatible with $\B$.
(6)~and (7)~deal with navigation along a binary relation, from
  the first to the second component in (6), and from the second to the
  first component in (7). In particular, (6)~states that the formula
  can quantify over the second component of a relation $R$ where $x$
  points to the first component if:
\begin{inparaenum}[\it (i)]
\item the second component of $R$ is a target role in $\B$,
  witnessing that $\Phi$ agrees with the unidirectional navigation
  imposed by $\B$ over $R$;
\item class $C$ belongs to the same hierarchy of the domain class for
  $R$, according to $\uml$;
\item $C'$ and $y$ are navigationally compatible with the inner
  formula $\Psi$, where $y$ is the newly quantified variable, and $C'$
  is the image class for $R$ according to $\uml$.
\end{inparaenum}
(7) works in a similar way, by simply inverting the second and first
components of $R$.
(8) next-state transition formulae are compatible if
  the class used in the guard belongs to the same hierarchy of $C$,
  and $C$ and $x$ are compatible with the inner subformula.

Notice that termination properties are always guaranteed to be
navigationally compatible with the corresponding \balsa model, since
$\cl{A}$ and $\tstate{\cl{A}}$ belong by definition to the same hierarchy.


Unfortunately, the following result shows that restricting \balsa
models to be unidirectional is
not sufficient to obtain decidability of checking termination properties.

\begin{theorem}
\label{thm:undec-unbounded}
Checking termination of unidirectional \balsa
models is undecidable.
\end{theorem}

\begin{proof}
Given a 2-counter machine
$\counter$, we produce a corresponding unidirectional \balsa model
$\B_\counter =
\tup{\uml^*,\emptyset,\set{\chart{\cl{2CM}}^*},\set{\proc{init}^*,\proc{run}^*}}$,
whose components are illustrated in Table~\ref{tab:2cm-unbounded}.
$\uml^*$ contains a single artifact $\cl{2CM}$, which can be
ready or halted, the latter being the termination state
($\tstate{\cl{2CM}} = \cl{Halted2CM}$), as attested by
$\chart{\cl{2CM}}^*$.
When the $\op{init}$ operation is applied, a new instance $m$ of
$\cl{Ready2CM}$ is created, attaching to it two dedicated objects of
type $\cl{Counter}$, using respectively role $c1$ and $c2$ of the
associations $hasC1$ and $hasC2$. Such  $\cl{Counter}$ objects mirror the two counters of
$\C$. In particular, each of the two $\cl{Counter}$ objects attached to $m$
has a 1-to-many association with $\cl{Item}$: at a given time,
the number of items attached to $m.c1$ ($m.c2$ resp.) represents the
value of the first (second resp.) counter in $\C$.

The artifact instance $m$ then  executes the process corresponding to
the $\op{run}$ event, which suitably encodes the program of $\C$:
\begin{inparaenum}[\it (i)]
\item incrementing the first counter translates into the inclusion of
  a new $\cl{Item}$ to the items of $m.c1$, i.e., to the set $m.c1.items$;
\item testing whether the first counter is 0 translates into checking
  whether set $m.c1.items$ is empty;
\item decrementing the first counter translates into the removal of
  one item from set $m.c1.items$ (it is not important which).
\end{inparaenum}
Table~\ref{tab:2cm-unbounded} shows how these three aspects can be
formalized in terms of activity diagrams and OCL queries
The management of the second
counter is analogous, with the only difference that it involves
$m.c2.items$ in place of $m.c1.items$. Figure~\ref{fig:unbounded-calc}
intuitively shows the evolution of a specific configuration of the
system in response to the application of two operations.

Observe that, as graphically depicted in $\uml^*$ (consistently with
the operations), $\B_\C$ is
unidirectional: all OCL expressions (except from that in $\op{init}$)
are navigational in $m$, and navigation unidirectionally flows from
$\cl{2CM}$ to $\cl{Counter}$ to $\cl{Item}$. Furthermore,
no two objects of type $\cl{Counter}$, nor two objects of type
$\cl{Item}$, are shared by different instances of $\cl{2CM}$. This
means that every instance of $\cl{Ready2CM}$ runs the process corresponding
to the program of $\C$ in total isolation with other
instances of $\cl{Ready2CM})$ and, consequently, either all halt or
none halt.
The claim follows by observing that $\C$ halts if and only
if all instances of $\cl{Ready2CM}$ eventually reaches the $\slabel{Halted2CM}$
state, i.e., properly terminate.
\begin{table*}[t!]
\begin{tabularx}{\textwidth}{|c|@{}C|}
\hline
$\uml^*$
&
$\chart{\cl{2CM}}^*$
\\
\hline
\tableg{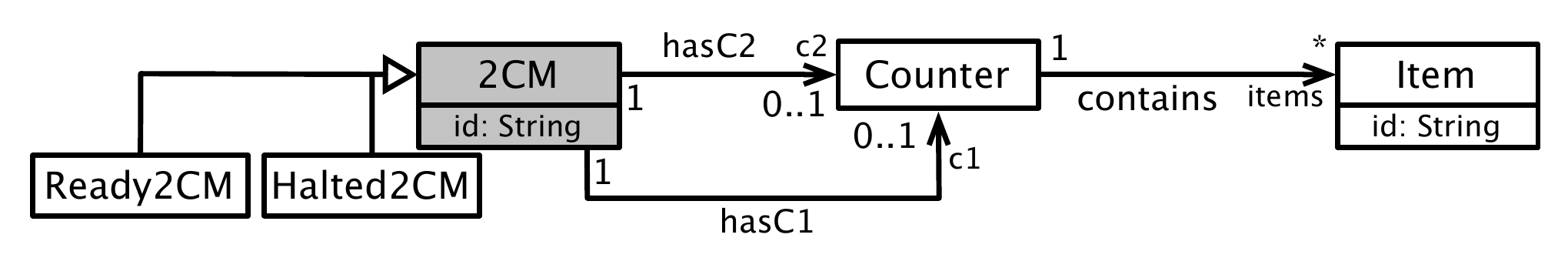}{9.5cm}
&
\tableg{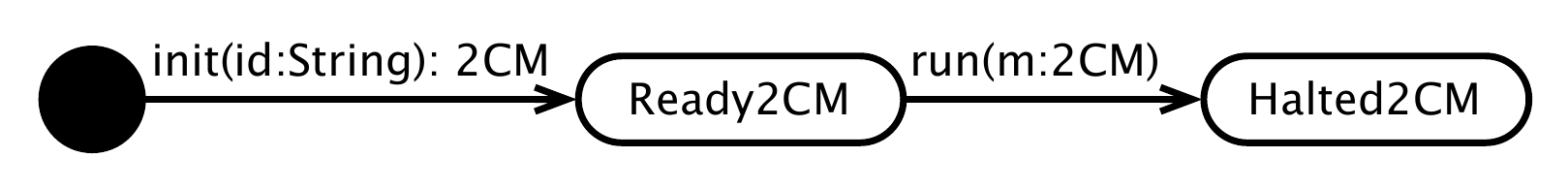}{7.5cm}
\\
\hline
\end{tabularx}

\begin{footnotesize}
\begin{tabularx}{\textwidth}{|l@{\ }|p{3.6cm}@{\ }|X|}
\hline
 $P_{init}^*$ &
\tableg{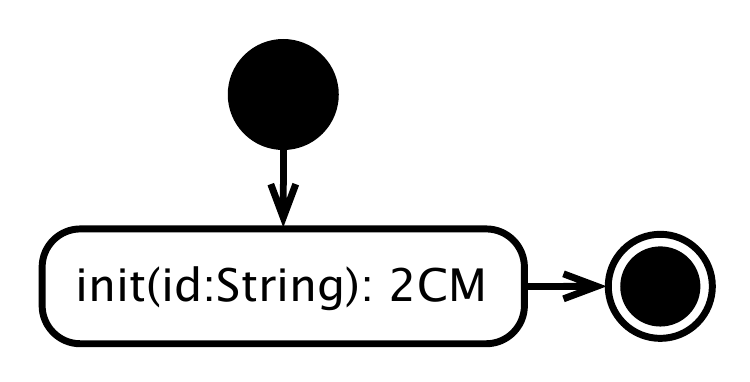}{3.6cm}
&
\vspace*{-.85cm}
$
\begin{array}[t]{@{}l@{\ }l}
 \textbf{pre:}  &\neg  (\cl{Ready2CM}.allInstances()\limp exists(m'|m'.id = id) ) \\
\textbf{post:} &
\oclq{ \cl{Ready2CM}.allInstances()\limp exists(m|
 \begin{array}[t]{@{}l}
m.oclIsNew() \land m.id = id \land result
 =m\\
{}\land (m.c1\limp exists(c_1|
c_1.oclIsNew()))\\
{}\land (m.c2\limp exists(c_2|c_2.oclIsNew()))~)\\
 \end{array}
} \\
\end{array}
$
 \\
\hline
\end{tabularx}
\end{footnotesize}

\begin{footnotesize}
\begin{tabularx}{\textwidth}{|l@{\ }|p{2.7cm}@{}|@{}p{4.5cm}|X|}
\hline
\multirow{4}{*}{$P_{run}^*$}
&
 start
&
\tableg{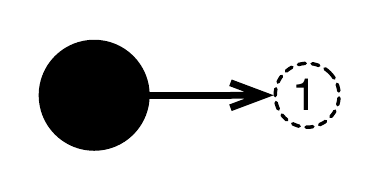}{1.6cm}
&
\\
\cline{2-4}
&
 $k: \inc{1}{k'}$
&
\tableg{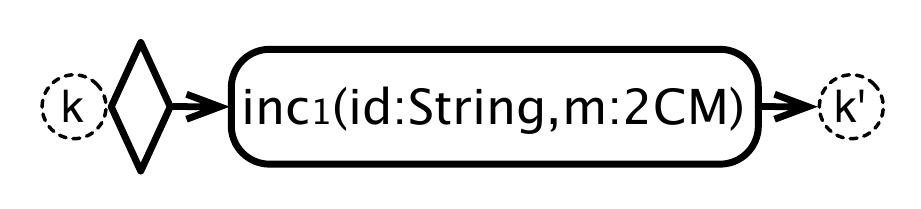}{4.4cm}
&
$
\op{inc_1}\
\begin{array}[t]{@{}l@{\ }l}
\textbf{pre:} & \oclq{\neg(m.c1.items\limp
  exists(i'|i'.id = id))}\\
\textbf{post:} & \oclq{m.c1.items\limp
  exists(i|i.oclIsNew() \land i.id = id)}
\end{array}
$
\\
\cline{2-4}
&
 $k: \cdec{1}{k'}{k''}$
&
\tableg{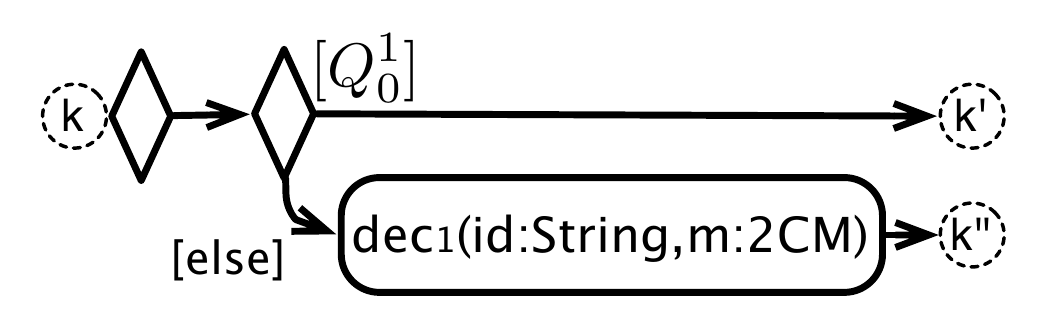}{4.8cm}
&
\vspace*{-.5cm}
$
\begin{array}{@{}l}
Q_0^1 = \oclq{m.c1.items \limp isEmpty()}\\[6pt]
\op{dec_1}\
\begin{array}[t]{@{}l@{\ }l}
 \textbf{pre:}  & \oclq{m.c1.items \limp
   exists(i | i.id = id)}\\
\textbf{post:} &\oclq{\neg(m.c1.items \limp
   exists(i | i.id = id))}
\end{array}
\end{array}
$
\\
\cline{2-4}
&
$n: \halt$
&
\tableg{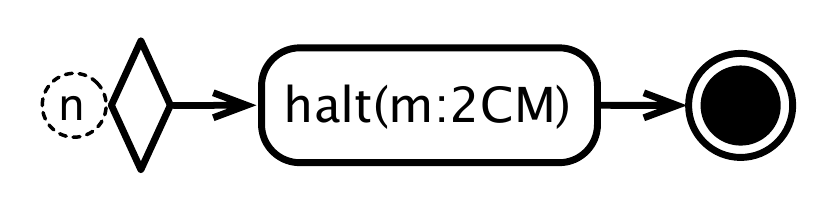}{4cm}
&
$
\op{halt} \
\begin{array}[t]{@{}l@{\ }l}
\textbf{post:} & \oclq{\neg(m.oclIsTypeOf(\cl{Ready2CM}))}\\
& {}\land \oclq{m.oclIsTypeOf(\cl{Halted2CM})}
\end{array}
$
\\
\hline
\end{tabularx}
\end{footnotesize}
\caption{Unidirectional \balsa model simulating a
  2-counter machine}
\label{tab:2cm-unbounded}
\end{table*}
\begin{figure*}[tb]
\centering
\includegraphics[width=\textwidth]{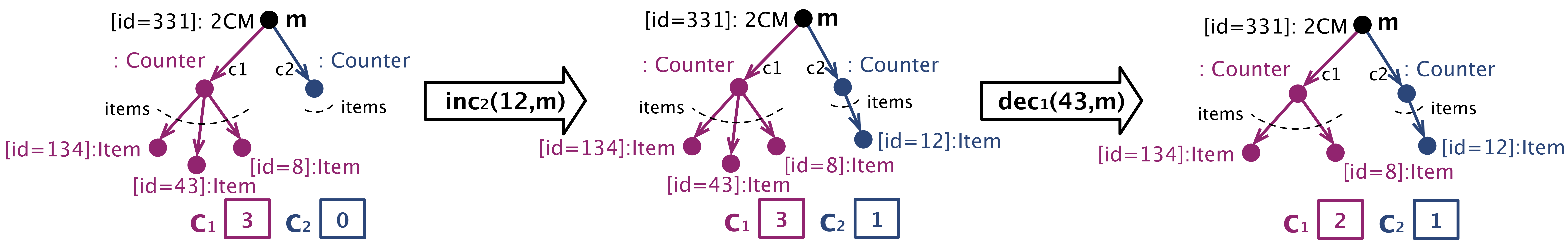}
\vspace*{-.4cm}
\caption{Sample counter manipulation using the \balsa model in Table~\ref{tab:2cm-unbounded}}
\label{fig:unbounded-calc}
\end{figure*}
\end{proof}


\subsection{ Cardinality-Bounded Models}
The source of undecidability in Theorem~\ref{thm:undec-unbounded}
relies in the \emph{contains} relation of $\uml^*$
(cf.~Table~\ref{tab:2cm-unbounded}), which relates its target role
\emph{items} with an unbounded cardinality. To overcome this issue, we
introduce the notion of cardinality-bounded \balsa model. A
\balsa model $\B = \tup{\uml,\ocl,\chartset,\procset}$ is \emph{cardinality-bounded} if $\B$ is navigational and each
target role in $\uml$ has a bounded cardinality, i.e., is associated
to a cardinality constraint whose upper bound is numeric.
$\B$ is \emph{N-cardinality-bounded} if the maximum upper bound
associated to a target role is $N$.
If there exists at least a target role with unbounded cardinality, i.e.,
associated to a cardinality constraint whose upper bound is $*$, then
$\B$ is instead said to be \emph{cardinality-unbounded}. Notice that
no cardinality restriction is imposed, for cardinality-bounded models,
on the cardinalities associated to roles that are not target roles.

With all these notions at hand, we are now able to state the main
result of this paper.

 \newenvironment{proofsk}{\noindent\emph{Proof (sketch).\ }}{\qed}

\begin{theorem}
\label{thm:dec}
Let $\B$ be an arbitrary unidirectional,
cardinality-bounded \balsa model. Verifying whether $\B$
satisfies a $\mulpers$
property navigationally compatible with $\B$ is decidable, and reducible to
finite-state model checking.
\end{theorem}
\begin{proof}
Let  $\B = \tup{\uml,\ocl,\chartset,\procset}$ be a cardinality-bounded, unidirectional \balsa
model, and let $\Phi$ be a $\mulpers$ property navigationally
compatible with $\B$.
On the one hand, by inspecting the notion of navigational compatibility, one can notice
that $\Phi$ is ``rooted'' in a single artifact class $\cl{S}$, subject to the
outermost subformula of the form $\exists x. S(x) \land \Psi(x)$ (or
$\forall x. S(x) \limp \Psi(x)$). Navigational compatibility then
ensures that $\Phi$ only mentions relations and classes that can be
reached by navigating $\uml$ using is-a relationships (in both
directions), or associations, in a direction that is compatible with
the unidirectionality imposed by $\B$.

On the other hand, as pointed out in Section~\ref{sec:nav}, in a
navigational model like $\B$ it is
impossible for artifact instances to share objects that belong to read-write
classes. This means that the evolution of an artifact instance is
completely independent from that of the other artifact
instances of the same type $\parentart{\cl{S}}$, or other artifact types.

By combining these two observations, we obtain that $\Phi$ obeys to a
sort of \emph{isolation property}:
\begin{compactitem}
\item $\Phi$ does not distinguish whether the system contains evolving artifact
  instances of types different than $\parentart{\cl{S}}$;
\item $\Phi$ does not distinguish whether the instances of
  $\parentart{\cl{S}}$ evolve in isolation, or co-evolve in a
  concurrent way.
\end{compactitem}
This isolation property is a data-aware variant of the free-choice
property of Petri nets. Thanks to such property, instead of directly considering the whole concurrent
evolution of the system, in which unboundedly many artifact instances
could be created over time and evolved in parallel, one can consider a
faithful, sound and complete abstraction of the system, which accounts
only for the concurrent evolution of those instances of type $\parentart{\cl{S}}$ present in the
  initial database of $\B$, plus
an additional artifact instance of type
  $\parentart{\cl{S}}$, nondeterministically
created and evolved in addition to the others.

Let $b_i$ be the number of artifact instances of type
$\parentart{\cl{S}}$ present in the initial database of the system.
From the fact that $\B$ is unidirectional and cardinality-bounded, we have
that each artifact instance can create only a bounded amount of objects
during its evolution. In fact, the number of objects that can be
created by an artifact instance is bounded by $(k \cdot N) ^ {l+1}$,
where:
\begin{inparaenum}[\it (i)]
\item $k$ is the number of relations in the schema (which bounds the
  number of relations that are collectively attached to an artifact/class in the schema),
\item $N$ is the maximum cardinality upper bound attached to a target role
belonging to a path rooted in $\parentart{\cl{S}}$, and
\item $l$ is the length of the
longest navigational path rooted in $\parentart{\cl{S}}$.
\end{inparaenum}
As a consequence, by considering the aforementioned
sound and complete abstraction, we have that at most $(b_i+1) \cdot N ^ {l+1}$ objects and artifact instances are simultaneously  present in a
system snapshot. The claim
then follows by:
\begin{inparaenum}[\it (i)]
\item applying the translation from \balsa models to
data-centric dynamic systems (DCDSs) \cite{Hariri2013}, provided in \cite{Estanol2013b};
\item  observing that the bound $(b_i+1) \cdot N ^ {l+1}$ implies that the
obtained DCDSs is state-bounded;
\item recalling that verification of $\mulpers$ properties over
  state-bounded DCDSs is decidable, and reducible to finite-state
  model checking \cite{Hariri2013}.
\end{inparaenum}
\end{proof}

\smallskip
An important open point is whether cardinality-boundedness is a
sufficient restriction for decidability per s\`e, i.e., without
necessarily imposing unidirectionality. The following theorem provides a
strong, negative answer
to this question, witnessing that both restrictions are
simultaneously required towards decidability.

\begin{theorem}
Checking termination of 1-cardinality-bounded, bidirectional \balsa
models is undecidable.
\end{theorem}

\begin{proof}
Given a 2-counter machine
$\counter$, we produce a corresponding 1-cardinality-bounded, bidirectional \balsa model
$\B_\counter =
\tup{\uml^b,\emptyset,\set{\chart{\cl{2CM}}^b},\set{\proc{init}^b,\proc{run}^b}}$,
whose components are illustrated in Table~\ref{tab:2cm-unbounded}.
$\uml^b$ contains a single artifact $\cl{2CM}$, which can be
ready or halted, the latter being the termination state
($\tstate{\cl{2CM}} = \cl{Halted2CM}$), as attested by
$\chart{\cl{2CM}}^b$.
When the $\op{init}$ operation is applied, a new instance $m$ of
$\cl{Ready2CM}$ is created, attaching a dedicated item that
represents the \emph{zero} point for both counters.

Intuitively, $m$ mirrors the two counters in $\C$ as follows. Thanks
to the fact that $m$ can navigate and manipulate the association \emph{hasNext} in
both directions (i.e., from left to right and from right to left), the
length of the right chain from the zero element $m.zero$ corresponds to the
value of the first counter, whereas the length of the left chain from
the zero element corresponds to the value of the second counter.

The artifact instance $m$ suitably encodes the commands in $\C$ as follows:
\begin{compactitem}
\item Incrementing the first counter requires to create a new
  $\cl{Item}$, and to put this object between the zero element and the
  old right-successor of it (cf.~$\op{inc}_1$, which conveniently
  exploits notation ``@pre'' to query the configuration of objects in
  the last predecessor state). This has the effect of increasing the
  length of the right chain of one unit.  The alternative operation
  $\op{incZ}_1$ handles the special case in which there is no
  right-successor from the zero element: in this case incrementing the
  counter just corresponds to add a new item on the right of the zero element.
\item Testing whether the first counter is 0 translates into checking
  whether set $m.zero.r$ is empty, i.e., whether it is true that the
  zero element does not have any right successor.
\item Decrementing the first counter translates into the removal of
  one item from set right chain of the zero element. There are two
  possible cases. In the first case, there is just a single
  right-successor, i.e., the counter has value 1. In this case,
  operation $\op{decS_1}$ just ensures that $m.zero.r$ does not have
  anymore this successor. If instead the right chain is longer than 1,
  then the decrement is handled by making the second right-successor
  of $m.zero$ the new direct right-successor of it, at the same time
  isolating the old direct right-successor.
\end{compactitem}
Table~\ref{tab:2cm-bidirectional} shows how these three aspects can be
formalized in terms of activity diagrams and OCL queries.
The management of the second
counter is analogous, with the only difference that it navigates the
left chain of the zero element, i.e., it exploits the $l$ role of
relation \emph{hasNext} in place of the $r$ role. Figure~\ref{fig:unbounded-calc}
intuitively shows the evolution of a specific configuration of the
system in response to the application of two operations.

Observe that, as clearly shown by $\uml^b$, $\B_\C$ is
1-cardinality-bounded, and is bidirectional, because relation
\emph{hasNext} is navigated on both directions, making both $l$ and
$r$ target roles. Furthermore, like for the reduction in
Theorem~\ref{thm:undec-unbounded}, each artifact instance is created
in state $\cl{Ready2CM}$, and evolves completely independently
from the other artifact instances.  This means that either all
instances of $\cl{Ready2CM}$ halt, or
none halt.
The claim follows by observing that $\C$ halts if and only
if all instances of $\cl{Ready2CM}$ eventually reach the $\slabel{Halted2CM}$
state, i.e., properly terminate.
\begin{table*}[t!]
\begin{tabularx}{\textwidth}{|c|@{}C|}
\hline
$\uml^b$
&
$\chart{\cl{2CM}}^b$
\\
\hline
\tableg{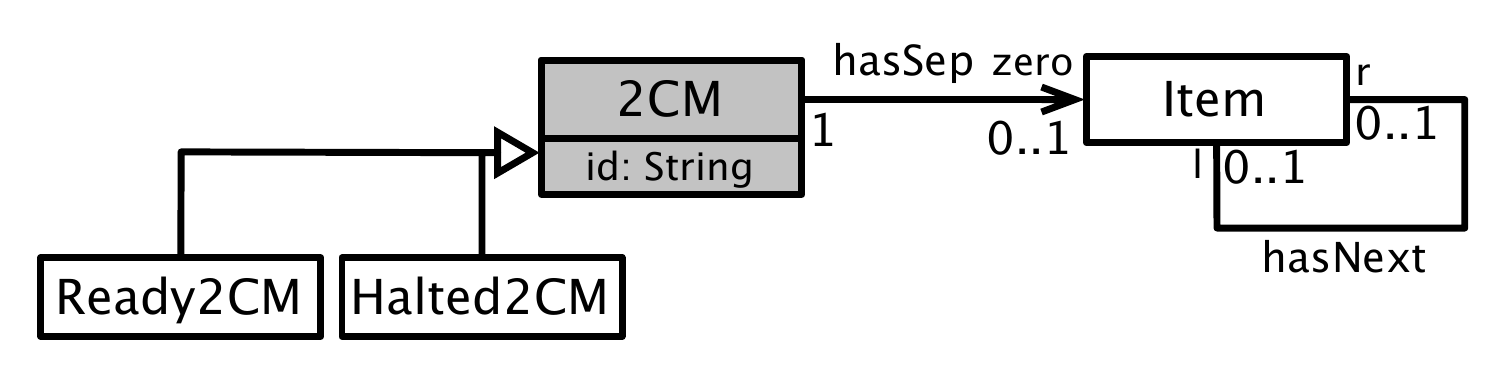}{8cm}
&
\tableg{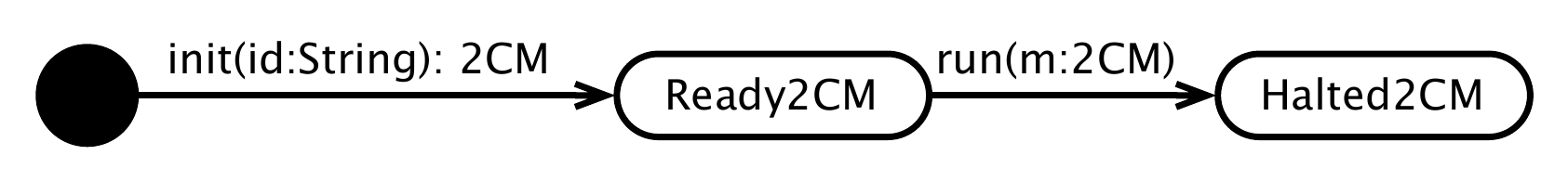}{8cm}
\\
\hline
\end{tabularx}

\begin{footnotesize}
\begin{tabularx}{\textwidth}{|l@{\ }|p{3.6cm}@{\ }|X|}
\hline
 $P_{init}^b$ &
\tableg{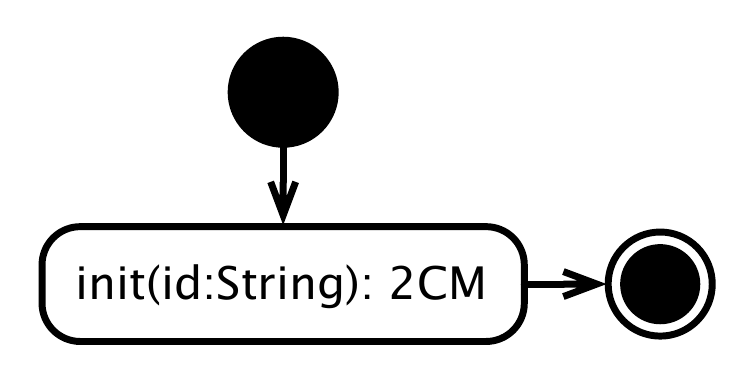}{3.8cm}
&
\vspace*{-.6cm}
$
\op{init}\
\begin{array}[t]{@{}l@{\ }l}
 \textbf{pre:}  & \neg  (\cl{Ready2CM}.allInstances()\limp exists(m'|m'.id = id) ) \\
\textbf{post:} &
\oclq{ \cl{Ready2CM}.allInstances()\limp exists(m|
\begin{array}[t]{@{}l}
m.oclIsNew() \land{} m.id = id \land result
 =m\\
\land m.zero \limp exists(i|i.oclIsNew())
~)
\end{array}
}\\
\end{array}
$
 \\
\hline
\end{tabularx}
\end{footnotesize}

\begin{footnotesize}
\begin{tabularx}{\textwidth}{|l@{\ }|p{2.7cm}@{}|@{}p{4cm}|X|}
\hline
\multirow{4}{*}{$P_{run}^b$}
&
 start
&
\tableg{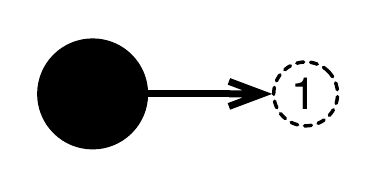}{1.6cm}
&
\\
\cline{2-4}
&
 $k: \inc{1}{k'}$
&
\tableg{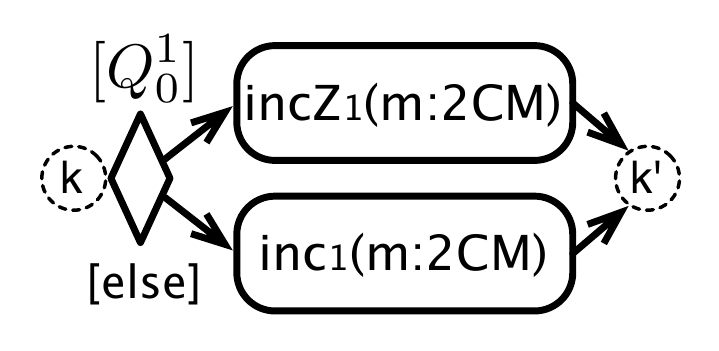}{4.2cm}
&
$
\begin{array}{@{}l}
Q_0^1 = (\oclq{m.zero.r \limp isEmpty()})\\[6pt]
\begin{array}{@{}ll@{\ }l}
\op{incZ_1}
& \textbf{post:} & \oclq{m.zero.r \limp exists(i|i.oclIsNew())}
\\
\op{inc_1}
& \textbf{post:} &  \oclq{m.zero.r \limp exists(i|}
\begin{array}[t]{@{}l}
\oclq{i.oclIsNew()}\\
\oclq{ \land  i.r = m.zero.r@pre)}
\end{array}
\end{array}
\end{array}
$
\\
\cline{2-4}
&
 $k: \cdec{1}{k'}{k''}$
&
\tableg{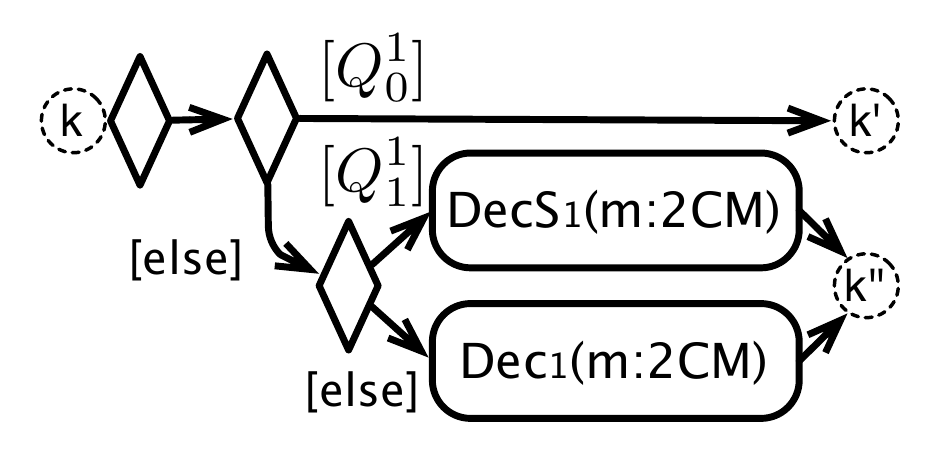}{4.2cm}
&
\vspace*{-.75cm}
$
\begin{array}{@{}l}
Q_0^1 =  (\oclq{m.zero.r \limp isEmpty()})
\quad
Q_1^1 =  (\oclq{m.zero.r.r \limp isEmpty()})\\[6pt]
\begin{array}{@{}ll@{\ }l}
\op{decS_1}
& \textbf{post:} & \oclq{m.zero.r \limp isEmpty()}
\\
\op{dec_1}
& \textbf{post:} &  \text{let } \oclq{old_r = m.zero.r@pre, new_r =  m.zero.r.r@pre }\\
&& \text{in } \oclq{m.zero.r = new_r \land (old_r.r \limp isEmpty())}
\end{array}
\end{array}
$
\\
\cline{2-4}
&
$n: \halt$
&
\tableg{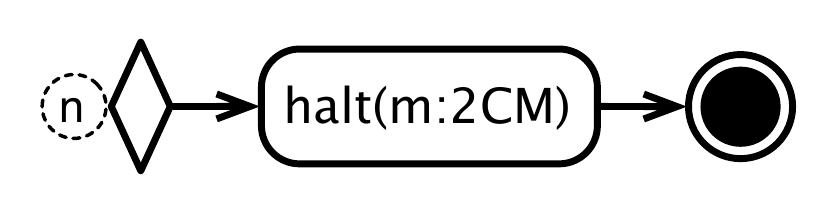}{4cm}
&
$
\op{halt} \
\begin{array}[t]{@{}l@{\ }l}
\textbf{post:} &  \oclq{\neg(m.oclIsTypeOf(\cl{Ready2CM}))}\\
&\oclq{\land m.oclIsTypeOf(\cl{Halted2CM})}
\end{array}
$
\\
\hline
\end{tabularx}
\end{footnotesize}
\caption{1-cardinality-bounded, bidirectional \balsa model simulating  a
  2-counter machine}
\label{tab:2cm-bidirectional}
\end{table*}
\begin{figure*}[t!]
\centering
\includegraphics[width=.8\textwidth]{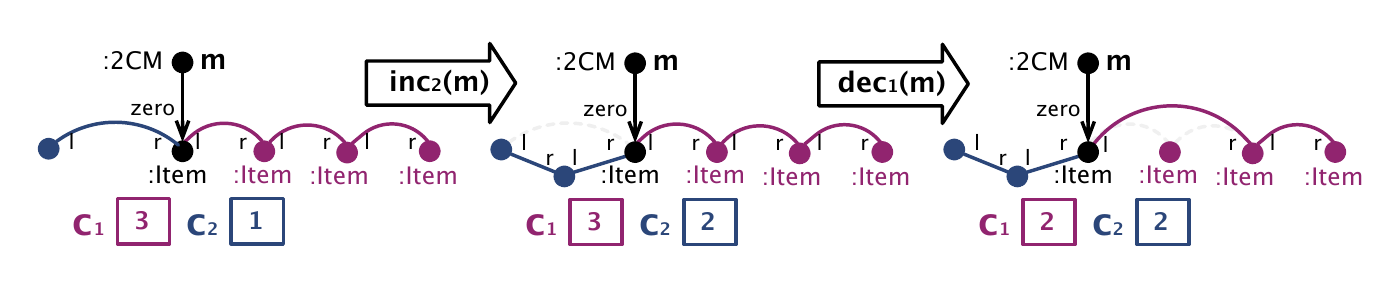}
\vspace*{-.4cm}
\caption{Sample counter manipulation using the \balsa model in Table~\ref{tab:2cm-bidirectional}}
\label{fig:bidirectional-calc}
\end{figure*}
\end{proof}

\begin{table*}[t!]
\begin{tabularx}{\textwidth}{|c|@{}C|}
\hline
$\uml^{bu}$
&
$\chart{\cl{Conn}}$
\\
\hline
\tableg{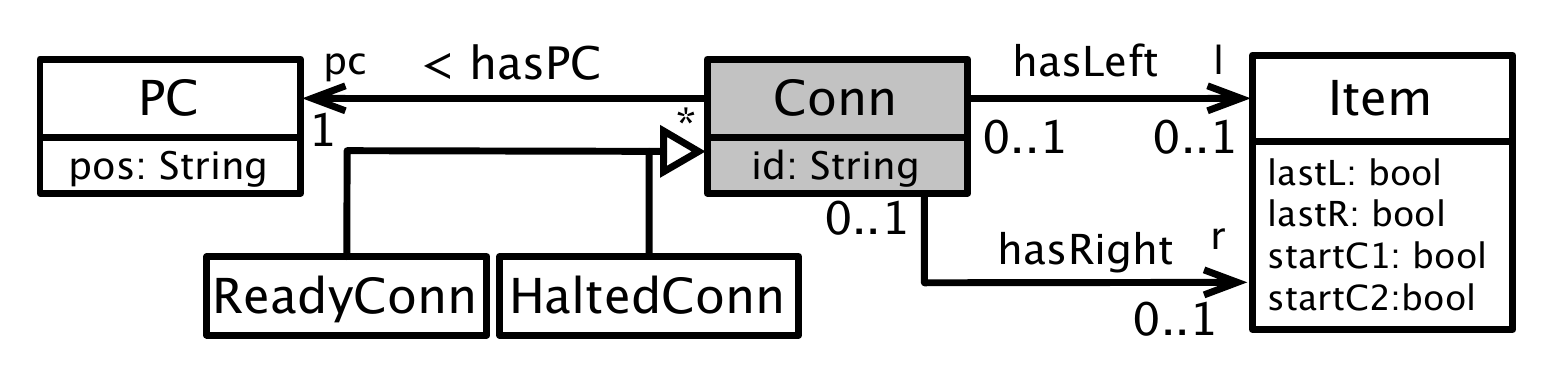}{8cm}
&
\tableg{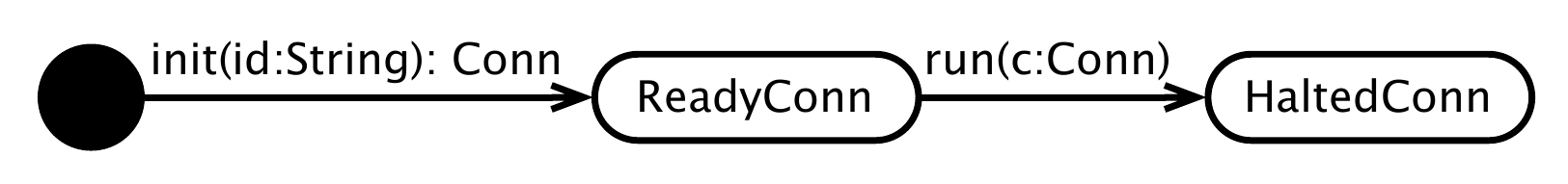}{8cm}
\\
\hline
\end{tabularx}

\begin{footnotesize}
\begin{tabularx}{\textwidth}{|l@{\ }|@{}p{3.5cm}@{\ }|X|}
\hline
 $P_{init}^{bu}$ &
 \tableg{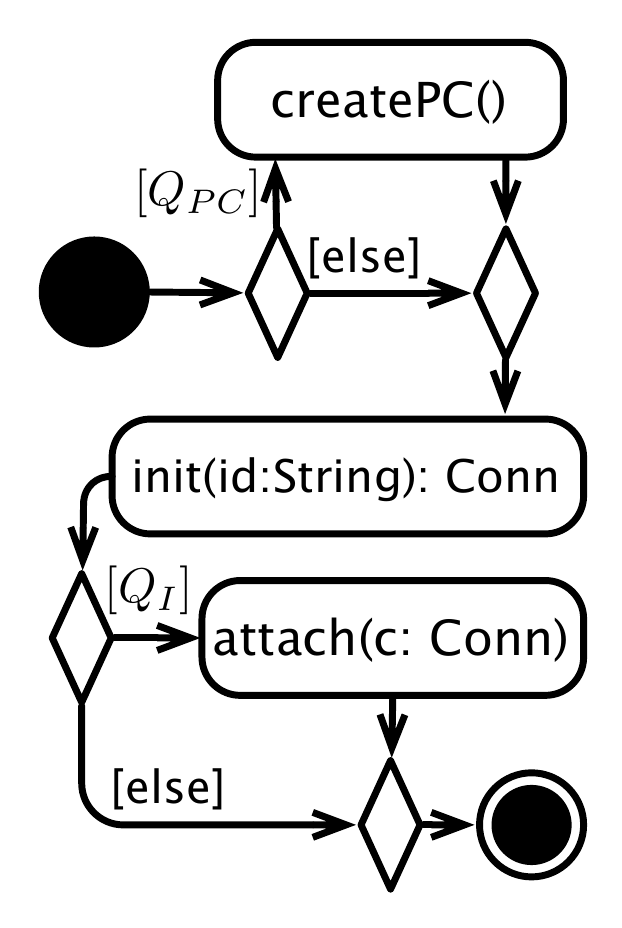}{3.5cm}
&
\vspace*{-1.6cm}
$
\begin{array}{@{}l}
\begin{array}{@{}l@{\ }l}
Q_{PC} & =  \oclq{\cl{PC}.allInstances() \limp isEmpty()}\\
Q_{I} & =  \oclq{\cl{Item}.allInstances() \limp isEmpty()}\\[6pt]
\end{array}
\\
\begin{array}[t]{@{}l@{\ }l@{\ }l}
\op{createPC}
& \textbf{post:} &
\oclq{\cl{PC}.allInstances() \limp exists(pc|pc.oclIsNew() \land
  pc.pos = \cval{1}) )}
\\
\op{init} &
 \textbf{pre:}  & \neg  (\cl{ReadyConn}.allInstances()\limp exists(m'|m'.id
 = id) ) \\
& \textbf{post:} &
\oclq{ \cl{ReadyConn}.allInstances()\limp exists(c|
\begin{array}[t]{@{}l}
c.oclIsNew() \land  c.id = id \land result = c \\
\land c.pc = (\cl{PC}.allInstances()))
\end{array}
}\\
\op{attach}
&
\textbf{post:}
&
\oclq{c.r \limp exists(s_{1} | s_1.oclIsNew() \land s_{1}.lastR = \cval{true} \land
  s_{1}.startC1 = \cval{true}) }\\
&&
\oclq{\land c.l \limp exists(s_2 | s_2.oclIsNew() \land s_2.lastL = \cval{true} \land s_2.startC2 = \cval{true}) }
\end{array}
\end{array}
$
 \\
\hline
\end{tabularx}
\end{footnotesize}

\begin{footnotesize}
\begin{tabularx}{\textwidth}{|l@{\ }|p{2.7cm}@{}|@{}p{4cm}|X|}
\hline
\multirow{4}{*}{$P_{run}^{bu}$}
&
 start
&
\tableg{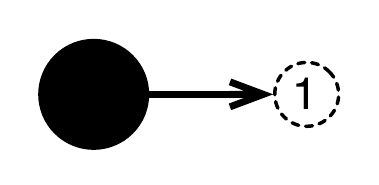}{1.6cm}
&
\\
\cline{2-4}
&
 $k: \inc{1}{k'}$
&
\tableg{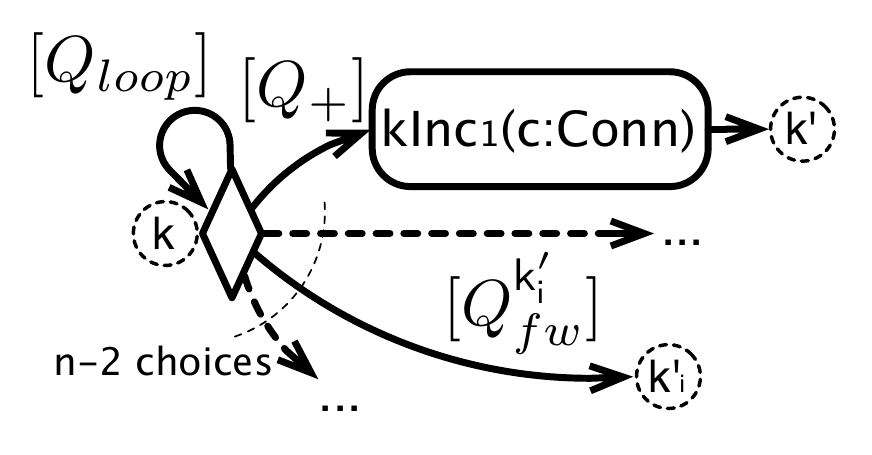}{4.2cm}
&
$
\begin{array}{@{}l}
\begin{array}{@{}l@{\ }c@{\ }l}
Q_{+} & = & \oclq{(c.pc = \cval{k})\land (c.l \limp
  isEmtpy()) \land (c.r \limp isEmpty())}\\
Q_{loop} & = & \oclq{(c.pc = \cval{k}) \land \neg((c.l \limp
  isEmtpy()) \land (c.r \limp isEmpty()))}\\
Q_{fw}^{\mathsf{k'_i}} & =  & (\oclq{c.pc = \cval{k'_i}})
\text{ for } \cval{k'_i} \in \set{\cval{1},\ldots,\cval{n}} \setminus \set{\cval{k},\cval{k'}}
\end{array}
\\[6pt]
\begin{array}[t]{@{}l@{\ }l@{\ }l}
\op{\cval{k}Inc_1}
& \textbf{post:} &
\text{let } \oclq{i = (\cl{Item}.allInstances() \limp select(i'|i'.lastR))}\\
&& \text{in } \oclq{c.l = i \land i.lastR = \cval{false}} \\
&& \phantom{\text{in }} \oclq{\land c.r
  \limp exists(i''|i''.isOclNew() \land i''.lastR = \cval{true}  )}\\
&& \phantom{\text{in }} \oclq{\land c.pc = \cval{k'} }
\\
\end{array}
\end{array}
$
\\
\cline{2-4}
&
 $k: \cdec{1}{k'}{k''}$
&
\tableg{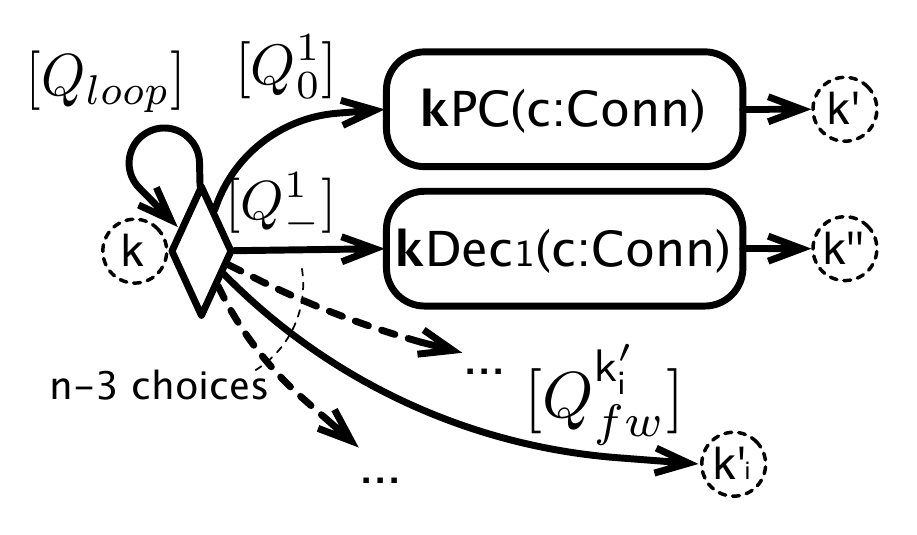}{4.2cm}
&
\vspace*{-1.3cm}
$
\begin{array}{@{}l}
\begin{array}{@{}l@{\ }c@{\ }l}
Q_0^1 & = & \oclq{(c.pc = \cval{k})\land c.r.lastR \land
  c.r.startC1}\\
Q_-^1 & = & \oclq{(c.pc = \cval{k})\land c.r.lastR \land \neg c.r.startC1}\\
Q_{loop} & = & \oclq{(c.pc = \cval{k}) \land \neg c.r.lastR}\\
Q_{fw}^{\mathsf{k'_i}} & =  & (\oclq{c.pc = \cval{k'_i}})
\text{ for } \cval{k'_i} \in \set{\cval{1},\ldots,\cval{n}} \setminus \set{\cval{k},\cval{k'},\cval{k''}}\\[6pt]
\end{array}
\\
\begin{array}[t]{@{}ll@{\ }l}
\op{\cval{k}PC}
& \textbf{post:} &
\oclq{ c.pc = \cval{k''} }
\\
\op{kDec_1}
& \textbf{post:} &
\text{let } \oclq{i_r = c.r@pre, i_l = c.l@pre}\\
&& \text{in } \oclq{i_l.lastR = \cval{true} \land i_r.lastR =
  \cval{false}}\\
&& \phantom{\text{in }}\oclq{ \land (c.l \limp isEmpty())
\land (c.r. \limp isEmpty())}
\end{array}
\end{array}
$
\\
\cline{2-4}
&
$n: \halt$
&
\tableg{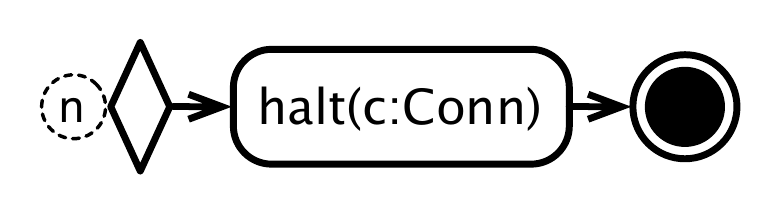}{4cm}
&
$
\op{halt} \
\begin{array}[t]{@{}l@{\ }l}
\textbf{post:} & \oclq{\neg(m.oclIsTypeOf(\cl{ReadyConn}))}\\
&\oclq{ \land m.oclIsTypeOf(\cl{HaltedConn})}
\end{array}
$
\\
\hline
\end{tabularx}
\end{footnotesize}
\caption{1-cardinality-bounded, unidirectional \balsa model with
  shared objects simulating a
  2-counter machine}
\label{tab:2cm-shared}
\end{table*}

\begin{figure*}[t!]
\centering
\includegraphics[width=\textwidth]{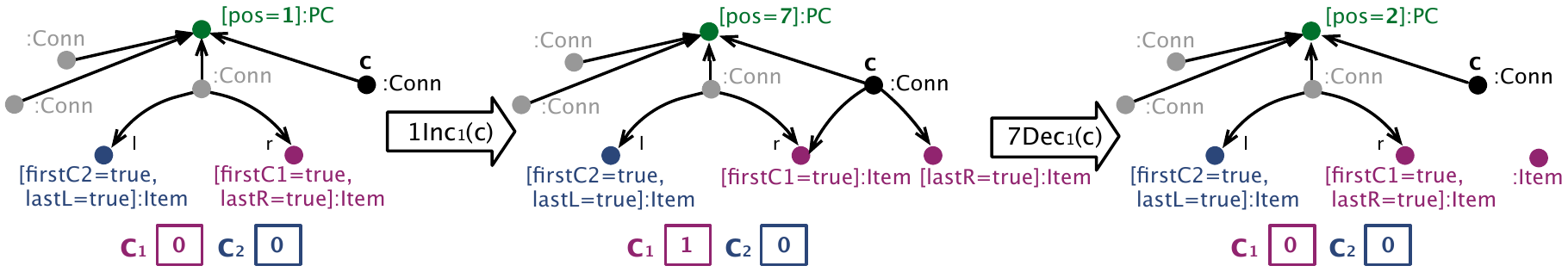}
\caption{Sample counter manipulation using the \balsa model in Table~\ref{tab:2cm-shared}}
\label{fig:shared-calc}
\end{figure*}

\subsection{Models With Shared Instances}
As argued in Section~\ref{sec:nav}, unidirectional \balsa models are
not able to make artifact instances share (read-write) objects. In this section, we
study what happens if we relax unidirectionality so as to support this
feature. A \emph{unidirectional \balsa model with shared instances}  $\B = \tup{\uml,\ocl,\chartset,\procset}$
 is a \balsa model in which, inside navigational expressions, it is possible to
add free queries over $\rw{\uml}$, provided that they \emph{do not
  contain} the expression $\oclq{oclIsNew()}$. Intuitively, this means
that new objects can only be created through standard navigational OCL
expressions, but at the same time it is possible to establish
associations with already existing objects that are not reachable by
simply navigating from the artifact instance.
The following theorem shows that this relaxation makes verification
again undecidable.
\begin{theorem}
Checking termination of 1-cardinality-bounded, unidirectional \balsa
models with shared instances is undecidable. 
\end{theorem}

\begin{proof}
Given a 2-counter machine
$\counter$, we produce a corresponding 1-cardinality-bounded
unidirectional \balsa model with shared instances
$\B_\counter =
\tup{\uml^{bu},\emptyset,\set{\chart{\cl{2CM}}^{bu}},\set{\proc{init}^{bu},\proc{run}^{bu}}}$,
whose components are illustrated in Table~\ref{tab:2cm-shared}.

As shown in Table~\ref{tab:2cm-shared}, $\uml^{bu}$ contains a single
artifact $\cl{Conn}$, which can be  ready or halted, the latter being the termination state
($\tstate{\cl{Conn}} = \cl{HaltedConn}$), as attested by
$\chart{\cl{Conn}}^{bu}$.
Due to cardinality boundedness and unidirectionality, a single
instance of $\B_\counter$ is not powerful enough to simulate
$\counter$. Hence, differently from the previous undecidability
proofs, the two counters are now simulated by unbouded chains of
artifact instances. In this light, the main difficulty is to properly
``synchronize'' such different instances so as to ensure that they
collectively implement the program of $\counter$, without interfering
with each other. To realize such a synchronization, all instances of
$\cl{Conn}$ share an instance of $\cl{PC}$, which represents a
``program counter'' to keep track of the current instruction to be
processed in $\counter$.  Intuitively, each instance of $\cl{Conn}$
represents a connection between two items; a chain of three items is
then built by using two instances of $\cl{Conn}$, making sure that the
first instance has on the right the same $\cl{Item}$ that the second
instance has on the left. This structure constitutes the basis for
simulating a counter.

Let us now go into the details of such a simulation. The
initialization transition in $\chart{\cl{Conn}}^{bu}$ consists now of
a complex activity diagram $P_{init}^{bu}$, which consists of the
following steps:
\begin{compactitem}
\item Initially, if there is no instance of the program counter, one
  instance is created, setting its ``position'' (represented by a
  string attribute \emph{pos}) to the constant string $\cval{1}$. If an
  instance of $\cl{PC}$ already exists, then this step is skipped.
\item The second step consists of the creation of a new connection artifact
  instance (of type $\cl{Conn}$), with a distinguished identifier. Upon
  creation, the $pc$ role of this connection points to the only available
  instance of $\cl{PC}$.
\item The third steps is applied only if no instance of class
  $\cl{Item}$ exists in the system. In this case, two special items
  are created so as to represent the zero elements for the two
  counters of $\counter$. This is done as follows:
\begin{compactitem}
\item The zero element for the first counter consists of a newly
  created instance $i_0^R$  of  $\cl{Item}$, whose boolean attribute $startC1$ is set to
  $\true$. Item  $i_0^R$ is attached to the right of the just created
  instance of $\cl{Conn}$. Since  $i_0^R$ is not on the left of any connection, also its boolean attribute $lastR$ is set to
  $\true$.
\item The zero element for the second counter consists of a newly
  created instance $i_0^L$  of  $\cl{Item}$, whose boolean attribute $startC2$ is set to
  $\true$. Item  $i_0^L$ is attached to the left of the just created
  instance of $\cl{Conn}$. Since  $i_0^L$ is not on the right of any connection, also its boolean attribute $lastL$ is set to  $\true$.
\end{compactitem}
\end{compactitem}
The structure obtained when $4$ instances of $P_{init}^{bu}$ are
executed in a row can be seen on the left of
Figure~\ref{fig:shared-calc}.

The idea behind the manipulation of counters starting from this
structure is to extend (resp., reduce) the chain on the right of item $i_0^R$ to
increment (resp., decrement) the first counter, and to extend (resp., reduce) the chain on the left of item $i_0^R$ to
increment (resp., decrement) the second counter. Since the case of the
second counter is obtained by just mirroring that of the first
counter, we just concentrate on the first counter.

The first important observation, which is common to the case of
counter increment and decrement, concerns the problem of
synchronization. On the one hand, as already pointed out we want all
instances of $\cl{Conn}$ to collectively realize the program of
$\counter$. On the other hand, there is no control on when new
instances of $\cl{Conn}$ are created. In particular, it could be the
case that a new connection is created when the other active
connections have already executed part of the program of
$\counter$. Similarly, since there is no control on how the different
active instances of $\cl{Conn}$ interleave with each other, when a
connection executes the portion of $P_{run}^{bu}$ corresponding to
instruction number $k$ in $\counter$, it must ensure that $k$ is
indeed the current intruction. More specifically, instruction number
$k$ always contains an initial choice, used to check whether the
program counter is indeed $k$ and, if so, whether the instance of
$\cl{Conn}$ that is executing the process is responsible for the
execution of instruction $k$, or should instead just execute an
``idle'' loop and wait that the responsible connection executes step
$k$. If the program counter stores in its $pos$ attribute an
instruction identifier different than $k$, then the process just
``jumps'' to the right step. If instead the program counter
corresponds to $k$, then a different behavior is exhibited depending
on whether the instruction number $k$ corresponds to an increment or
conditional decrement for the first counter.

In the case of increment:
\begin{compactitem}
\item If the connection is not associated to any item on its left and
  its right (i.e., it is not part of any chain), then the connection
  becomes responsible for the increment, which is atomically executed
  using the operation $\textbf{k}Inc_1$. The increment is realized as
  follows:
\begin{compactitem}
\item The unique item (called $i$) that has attribute $lastR$ set to
  $\true$ is selected.
\item This item is attached on the left of the current connection,
  setting its $lastR$ attribute to $\false$. In this way, it is easy
  to see that an item has $lastR = \false$ if and only if there is no
  connection that has it on the left.
\item A new item is created and attached on the right of the current
  connection, setting its $lastR$ attribute to $\true$. This newly
  created item represents the increment of the first counter, and the
  current connection acts as the last connection of the chain
  simulating the first counter.
\item The program counter is updated, setting its $pos$ attribute to
  the string that corresponds to the new instruction identifier
  $k'$. Since $k'$ is a pre-defined string, each increment is
  different from the others, and this is why each specific increment
  is mapped to a separate operation in $P_{run}^{bu}$.
\end{compactitem}
Considering e.g., the case of instruction $1: \inc{1}{7}$, the
central part of Figure \ref{fig:shared-calc} represents the new data
configuration after the execution of this step by one of the
connections that are currently active but not associated to any item.
\item If instead the connection is already attached to an item on the
  left or on the right, then it executes an idle step, going back to
  check whether the program counter is still $k$ or has instead been updated.
\end{compactitem}

In the case of conditional decrement:
\begin{compactitem}
\item If the connection has on its right an item whose attribute
  $lastR$ is $\true$, then the connection becomes responsible for the
  conditional decrement. Two cases may then arise: either the first
  counter is 0, and consequently only the program counter must be
  updated, or the counter is positive, and consequently the counter
  must be decremented before updating the program counter. The test
  for zero can be easily captured in $\B^{bu}$ by testing whether the
  item having $lastR = \true$ also has $startC1 = \true$: if so, then
  the first counter is zero, if not, then the first counter is
  positive. In the former case, captured by query $Q^1_0$, the
  specific task $\textbf{k}PC$ is executed, whose effect is simply to
  update the attribute $pos$ of the program counter to the string
  corresponding to $k''$; since $k''$ is a pre-defined string, each
  program counter update is different from the others, and this is why
  each specific program counter update is mapped to a separate
  operation in $P_{run}^{bu}$. In the latter case, captured by query
  $Q^1_-$, an atomic decrement and program counter update is executed using
the operation $\textbf{k}Dec_1$. The decrement is realized as
  follows:
\begin{compactitem}
\item The item that was previously on the right of the connection is
  updated making its $lastR$ attribute equal to $\false$.
\item The item that was previously on the left of the connection
  (i.e., on the right of the previous connection along the chain) is
  updated making its $lastR$ attribute equal to $\true$.
\item The connection is disconnected from both such items, hence
  reducing the chain of one item. This has also the indirect effect of
  making the connection eligible for being responsible of a successive
  increment.
\item The program counter is updated, setting its $pos$ attribute to
  the string that corresponds to the new instruction identifier
  $k'$. Since $k'$ is a pre-defined string, each decrement is
  different from the others, and this is why each specific decrement
  is mapped to a separate operation in $P_{run}^{bu}$.
\end{compactitem}
Considering the case of instruction $7: \cdec{1}{2}{9}$, the
right part of Figure \ref{fig:shared-calc} represents the new data
configuration after the execution of this step by the
connection that is currently at the end of the right chain.
\item If instead the connection does not have on its right the element
  whose $lastR$ attribute is $\true$, then it executes an idle step, going back to
  check whether the program counter is still $k$ or has instead been updated.
\end{compactitem}
As soon as one of the active connection artifact instances sets the
program counter to the constant $\cval{n}$, all active connections
move to the final part of $P_{run}^{bu}$, where they are moved from
the $\cl{ReadyConn}$ to the $\cl{HaltedConn}$ state. If new instances
of $\cl{Conn}$ are subsequently created, they immediately jump to
execute this task as well (in fact, they all share the same program
counter, whose $pos$ attribute continues to be $\cval{n}$).
 This means that either all
instances of $\cl{ReadyConn}$ halt, or
none halts.
The claim follows by observing that $\C$ halts if and only
if all instances of $\cl{ReadyConn}$ eventually reach the $\slabel{HaltedConn}$
state, i.e., properly terminate.
\end{proof}

\smallskip
We close this thorough analysis by showing that, if we introduce a
bound on the number of artifact instances that are simultaneously
active in the system, verification becomes decidable for this specific
class of \balsa models. This technique cannot be applied
to unrestricted nor unbounded \balsa models: by inspecting the proofs of Theorems~\ref{thm:undec-unrestricted}
and~\ref{thm:undec-unbounded}, one can easily notice that
undecidability holds even when there is just a single active artifact
instance.
\begin{theorem}
\label{thm:dec-bounded}
Verification of $\mulpers$ properties over cardinality-bounded, unidirectional \balsa
models with shared instances of read-write classes is decidable and reducible to finite-state model checking when the number of simultaneously active artifact
instances is bounded.
\end{theorem}
\begin{proof}
Let $\B$ be a cardinality-bounded, unidirectional \balsa
model. By combining unidirectionality and cardinality-boundedness, we have
that an artifact instance can create only a bounded amount of objects
during its evolution. In fact, the number of objects that can be
created is bounded by $(k\cdot N)^{l+1}$, where $k$, $N$ and $l$ are
as in the proof of Theorem~\ref{thm:dec}. 
Since the number of simultaneously active artifact
instances is bounded, say, by a number $b$, then at each time point
the number of objects and artifact instances present in the overall
system is bounded by $b \cdot (k\cdot N)^{l+1}$. The claim
then follows by:
\begin{inparaenum}[\it (i)]
\item applying the translation from \balsa models to
DCDSs, described in \cite{Estanol2013b};
\item  observing that the bound $b \cdot (k\cdot N)^{l+1}$ implies that the
obtained DCDS is state-bounded;
\item recalling that verification of $\mulpers$ properties over
  state-bounded DCDSs is decidable, and reducible to finite-state
  model checking \cite{Hariri2013}.
\end{inparaenum}
\end{proof}
It is important to observe that bounding the number of
simultaneously active artifact instances still allows one to create an
unbounded amount of artifact instances over time, provided that they
do not accumulate in the same snapshot. In this light,
Theorem~\ref{thm:dec-bounded} closely resembles the result given in
\cite{Solomakhin2013} for business artifacts specified in the GSM notation.

To show the practical relevance of these results, we return to
our example, presented in Section~\ref{sec:example}. It is a realistic
example of a data-centric business process.  At the same time it
is a cardinality-bounded, unidirectional model with shared instances
coming from a read-only relation (\cl{ItemType}). Hence, it falls into
the case of Theorem~\ref{thm:dec}, for which verification is
decidable even in presence of unboundedly many simultaneously active artifact
instances. In the case where artifacts share a read-write relation,
decidability requires an additional bound on the
number of simultaneously active artifact instances, so
as to fall into Theorem~\ref{thm:dec-bounded}.




\section{Related Work}

This section will examine alternative representations for artifact-centric business process models, with the focus on the data dimension. In those cases where it is possible, we will review the decidability results that have been obtained for the formal verification of these models. However, most of these results are applicable to models grounded on logic or mathematical notations that do not provide a practical business level representation.  We will first begin by looking at alternative graphical representations and we will continue with alternatives grounded on logic.

Apart from the work in \cite{Estanol2013} that we have considered in this paper, there are also other approaches that use UML class diagrams to represent the data dimension, such as \cite{Fahland2011}. However, \cite{Fahland2011} turns to proclets (a labeled Petri net with ports) to represent the internal lifecycle of the artifact and how it relates to other artifacts. 

ER models \cite{Bhattacharya2009} are similar to UML class diagrams as they also allow representing the relationships between the artifacts and their attributes. The PHILharmonic Flows framework \cite{Kunzle2011} represents business processes with data in a graphical way, using a model which falls in-between a UML diagram and a database schema representation. Unlike our approach, it does not distinguish between what we call business artifacts and objects. 

Another alternative is to extend BPMN to allow the representation of data-dependencies in the business process model \cite{Meyer2013}. However, \cite{Meyer2013} does not have a specific diagram showing the relationships between the data or artifacts. The Guard-Stage-Milestone (GSM) approach \cite{Hull2011b,Damaggio2013} represents the artifact and its lifecycle in one model, which shows the guards, stages and milestones involved in the evolution of an artifact. In contrast to the UML class diagram, GSM does not show graphically the relationships between the artifacts: they are encoded as attributes instead.

Several works deal with the formal verification of GSM models and study their decidability. For instance, \cite{Solomakhin2013} uses an approach that is very close to ours. It relies on the notion of state-boundedness to guarantee decidability. Similarly, \cite{Belardinelli2012} deals with decidability of GSM models but taking agents (i.e. users or automatic systems) into consideration. \cite{Gonzalez2012} also applies model checking to these models, but its implementation restricts the data types and only admits one artifact instance. Both \cite{Belardinelli2012} and \cite{Gonzalez2012} use CTL or a variant of CTL, neither of which are as powerful as $\mu$-calculus.

There are several works \cite{Hariri2013,Hariri2013b,Calvanese2012} that deal with verification of artifact-centric business process models represented by means of a data-centric dynamic system (DCDS). DCDSs are grounded on logic. \cite{Hariri2013} represents artifacts by means of a relational database schema, \cite{Hariri2013b} uses a knowledge and action base defined in a variant of Description Logics, and \cite{Calvanese2012} maps an ontology to a DCDS. All these works define the properties to be checked in variants of $\mu$-calculus. They ensure decidability either by state-boundedness \cite{Hariri2013,Calvanese2012} or by limiting the calls to functions that obtain new values \cite{Hariri2013b,Hariri2013}.

Works such as \cite{Damaggio2012} and \cite{Gerede2007} also verify the fulfillment of properties by the model but they both define properties in variants of LTL or CTL (respectively), making them less powerful than $\mu$-calculus. \cite{Damaggio2012} represents the data by means of a database schema. It allows the use of integrity constraints in the data and arithmetic operations, requiring the condition of feedback-freedom (i.e. output variables cannot be reused from one function to the next) to guarantee decidability. \cite{Gerede2007} opts for bounding the domain values or to limit the language that is used instead. Artifacts are represented by means of a tuple which includes a set of attributes.

\section{Conclusions}

We have analyzed the decidability of verification for artifact-centric business process models defined according to the BALSA framework and at a high level of abstraction. That is, we have lifted the decidability conditions from the formal, low-level representations, to the business level, to establish conditions which can be considered by the modeler of the process. Although we have focused on the representation of these elements using UML, our results could be extended to other forms of representation.

As a result of our analysis, we have concluded that verification of artifact-centric process models is only decidable when:
\begin{inparaenum}[\it (i)]
\item artifacts are linked to a bounded number of objects,
\item two different artifacts only share read-only objects,
\item expressions in the pre and postconditions of the operations are navigational starting from the artifact instance being manipulated, and
\item the associations specified among two classes are not navigated back and forth.
\end{inparaenum}
If any of these four conditions is relaxed, then we end-up with undecidability. Regaining decidability when the model contains shared read-write objects requires to put a bound on the number of simulatenously active artifact instances. Although these conditions are restrictive, they still allow for the definition of relevant situations in practice.

As further work, we would like to pursue this line of research so as to characterize concrete, real-life settings
for which decidability of verification is guaranteed. We also plan to provide a more fine-grained characterization of how read and write operations might interact without undermining decidability.
Finally, we aim at studying the practical applicability of our verification techniques, by understanding how the exponentiality in the data that is inherent in data-aware systems can be tamed, through a suitable modularization/partitioning of the data into independent portions.

\section*{Acknowledgments}
This research has been partially supported by the EU FP7 IP project
Optique (\emph{Scalable End-user Access to Big Data}),
grant agreement n.~FP7-318338, MICINN projects TIN2011-24747 and TIN2008-00444, Grupo Consolidado,
the FEDER funds and Universitat Polit\`{e}cnica de Catalunya.

\bibliographystyle{abbrv}
\bibliography{references}

\balancecolumns

\end{document}